\documentclass[reqno]{amsart}
\usepackage[lite]{amsrefs}
\usepackage{amssymb}
\usepackage{bm}
\usepackage[all,cmtip]{xy}

\usepackage{hyperref}

\usepackage[hmargin= 1.3 in,vmargin=1 in]{geometry}

\newcommand{\T}{\mathbb{T}}

\newcommand{\R}{\mathbb{R}}
\newcommand{\C}{\mathbb{C}}
\newcommand{\Z}{\mathbb{Z}}
\newcommand{\N}{\mathbb{N}}

\newcommand{\eps}{\varepsilon}
\newcommand{\ii}{\mathrm{i}}
\newcommand{\mcV}{\mathcal{V}}
\newcommand{\vphe}{\vph_\eps}
\newcommand{\vphes}{\vph^*_\eps}

\newcommand{\mcN}{\mathcal{N}}

\newcommand{\mcA}{\mathcal{A}}
\newcommand{\mcF}{\mathcal{F}}

\newcommand{\mfh}{\mathfrak{h}}

\newcommand{\mcS}{\mathcal{S}}
\newcommand{\mbt}{\mathbf{t}}
\newcommand{\mbx}{\mathbf{x}}
\newcommand{\mby}{\mathbf{y}}

\newcommand{\vph}{\varphi}
\newcommand{\WW}{\mathrm{W}}

\newcommand{\dd}{\mathrm{d}}

\newcommand{\e}{\mathrm{e}}

\newcommand{\mfBp}{\mathfrak{B}_p}
\newcommand{\mfhp}{\mathfrak{h}^{(p)}}
\newcommand{\mcCp}{\mathcal{C}_p}
\newcommand{\mc}{\mathcal}
\newcommand{\mf}{\mathfrak}
\newcommand{\om}{\omega}
\newcommand{\Th}{\Theta}

\newcommand{\mbs}{\mathbf{s}}

\newcommand{\mcL}{\mathcal{L}}

\numberwithin{equation}{section}

\theoremstyle{plain}

\theoremstyle{definition}

\theoremstyle{remark}

\usepackage{graphicx, color}

\definecolor{darkred}{rgb}{0.9,0,0.3}
\definecolor{darkblue}{rgb}{0,0.3,0.9}

\newcommand{\vertiii}[1]{{\left\vert\kern-0.25ex\left\vert\kern-0.25ex\left\vert #1 
		\right\vert\kern-0.25ex\right\vert\kern-0.25ex\right\vert}}

\theoremstyle{plain} 
\newtheorem{theorem}{Theorem}[section]
\newtheorem*{theorem*}{Theorem}
\newtheorem{lemma}[theorem]{Lemma}
\newtheorem*{lemma*}{Lemma}
\newtheorem{corollary}[theorem]{Corollary}
\newtheorem*{corollary*}{Corollary}
\newtheorem{proposition}[theorem]{Proposition}
\newtheorem*{proposition*}{Proposition}

\newtheorem*{conjecture*}{Conjecture}

\theoremstyle{definition} 
\newtheorem{definition}[theorem]{Definition}
\newtheorem*{definition*}{Definition}

\newtheorem*{example*}{Example}
\newtheorem{remark}[theorem]{Remark}
\newtheorem*{remark*}{Remark}
\newtheorem{assumption}[theorem]{Assumption}
\newtheorem*{assumption*}{Assumption}

\begin{document}
\title[Derivation of focusing $\Phi_1^6$ with rough cut-off]{A perturbative microscopic derivation of the focusing $\Phi^6_1$ measure with rough cut-off}
\author{Shahnaz Farhat}
\address{Department of Mathematics, University of T\"ubingen, 72076 T\"ubingen, Germany.}
\email{shahnaz.farhat@uni-tuebingen.de}
\author{Andrew Rout}
\address{Dipartimento di Matematica, Politecnico di Milano, P.zza Leonardo da Vinci 32, 20133, Milano, Italy.}
\email{andrewjames.rout@polimi.it}

\begin{abstract}
We give a derivation of the Gibbs measure for the focusing nonlinear Schr\"odinger equation (NLS) on the circle with rough cut-off. This extends earlier work by Sohinger and the second author, which proved analogous results for smooth cut-offs. Our proof is based on the perturbative expansion developed by Fr\"ohlich, Knowles, Schlein, and Sohinger in \cite{FKSS17}, and provides an alternative proof of the recent derivation given in \cite{LNZ26} by L\"u, Nam, and Zhu. To prove convergence of the explicit terms, we employ a Wigner measure approach and an inductive argument to overcome the lack of smoothness for the cut-off. In particular, we give a derivation of the Gibbs measure for the focusing quintic NLS with the optimal cut-off from \cite{OST22}.
\end{abstract}
\maketitle
\section{Introduction}
Suppose that $\mfh$ is a Hilbert space and $H $ is an associated Hamiltonian function. The {\it Gibbs measure} is a generalisation of the canonical ensemble that is a probability measure on $\mfh$ formally defined by
\begin{equation}
\label{formal_defocusing_Gibbs}
\dd \mu^{\mathrm{Gibbs}}(\vph) := \frac{1}{z_{\mathrm{Gibbs}}} \e^{-H(\vph)} \, \dd \varphi.
\end{equation}
Here $z_{\mathrm{Gibbs}}$ is the normalisation constant, or {\it partition function}, that  makes $\dd \mu^{\mathrm{Gibbs}}$ a probability measure, and $\dd \varphi$ is the formally defined Lebesgue measure on $\mfh$ (which is ill-defined when $\mfh$ is infinite dimensional). The construction of such measures was rigorously studied in constructive quantum field theory, see for example \cite{GJ81,Nel73,Nel73b,Sim74}. These measures were further studied in \cite{BG20,BTT13,BS96,CFL16,GH21,LRS,OST22} and the references within. The study of such measures as invariant measures for Hamiltonian PDEs was initiated by Bourgain in \cite{Bou94}, with some preliminary results known by Zhidkov \cite{Zhi91}, and further studied in \cite{Bou96,Bou97}. The invariance of such measures can be used to construct global solutions to PDEs with rough initial conditions. For expository works on this subject, we direct the reader to \cite{BTT18,NS19,OT18}. We also direct the reader to \cite{AFS24,AS23,Bri22,Bri24,BDNY24,DNY22,DNY24,DR23,OTT24,RSTW25} and the references within for more recent developments.

If $H$ is not positive definite, the Gibbs measure may not always be normalisable (since $z_\mathrm{Gibbs}$ may be infinite). When $H$ is not positive definite, we will say that we are in the {\it focusing} case. In this case one considers the modification given by
\begin{equation}
\label{focusing_Gibbs_formal}
\dd \mu_{\mathrm{Gibbs}}^\chi(\vph) := \frac{1}{z^\chi_{\mathrm{Gibbs}}} \e^{-H(\vph)} \chi(\|\vph\|_{L^2(\mfh)}^2)\, \dd \varphi,
\end{equation}
where $\chi$ is a suitable cut-off function and $z^\chi_{\mathrm{Gibbs}}$ is the normalisation constant which makes \eqref{focusing_Gibbs_formal} a probability measure. Measures of this form were studied in \cite{ARS24,Bou94,Bou96,Bou97,LRS,LW22,Nik25,OST22,OTT24,RSTW25} and the references therein. We also mention \cite{DR23,DRTW23,OQ13}, which studied the canonical ensemble, where one conditions on the mass of the system.

Of particular interest in this paper is the (local) focusing quintic nonlinear Schr\"odinger equation in one dimension, which exhibits a phase-transition. In particular, the measure \eqref{focusing_Gibbs_formal} can only be constructed for cut-offs with sufficiently small support, with blow-up for larger ones, see \cite{DRTW23,CFL16,LRS,OST22,RSTW25}. Therefore it is of particular interest in this setting to study \eqref{focusing_Gibbs_formal} with an indicator function $\chi \equiv \chi_{[0,K_{\mathrm{max}}]}$, where $K_\mathrm{max}$ is the maximal support for which there is not blow-up.

In this paper, we will study how the Gibbs measure given by \eqref{focusing_Gibbs_formal} associated to the one-dimensional NLS arises from many-body quantum mechanics. In particular, we consider the case where the cut-off in \eqref{focusing_Gibbs_formal} is an indicator function, and we can consider cut-offs up to the critical cut-off $K_{\mathrm{max}}$. This problem has been studied for measures of the form of \eqref{formal_defocusing_Gibbs} in \cite{FKSS17,FKSS18,FKSS20,FKSS22,FKSS23,FKSS22b,JR25,LNR15,LNR18,LNR21,NZZ25,Soh19}, and for measures with truncation as in \eqref{focusing_Gibbs_formal} in \cite{DR24,RS22,RS23}. Our result extends earlier work by the second author and Sohinger in \cite{RS22,RS23}, which considered the same problem with a smooth cut-off function, and gives an alternative derivation to the recent work by L\"u, Nam, and Zhu \cite{LNZ26}; see Remark \ref{LNZ_remark}.

\subsection{Setting and assumptions}
\label{section_setup_quintic}
In this section we clarify our setting. Set $\Lambda := \T \equiv \R/\Z = [-\frac{1}{2},\frac{1}{2})$. We define $\mfh := L^2(\Lambda) \equiv L^2$. We define the inner product on $\mfh$ as
\begin{equation*}
\langle f,g \rangle := \int_\mfh \dd x \, \overline{f(x)}g(x).
\end{equation*}
We fix the following assumption.
\begin{assumption}
\label{free_Hamiltonian_assumption}
Let $-\Delta$ denote the Laplacian on $\mfh$. Then we consider $h := -\Delta + 1$.
\end{assumption}
We note that in Assumption \ref{free_Hamiltonian_assumption}, one has that $h$ is a strictly positive self-adjoint densely-defined operator on $\mfh$. Moreover, one can write $h$ spectrally as
\begin{equation*}
\sum_{k \in \N} \lambda_k u_k u_k^*,
\end{equation*}
where $\lambda_k$ is the sequence of eigenvalues of $h$, and $u_k$ the corresponding eigenfunctions. We note that under Assumption \ref{free_Hamiltonian_assumption}, one has
\begin{equation*}
\mathrm{Tr}(h^{-p}) < \infty 
\end{equation*}
for any $p > \frac{1}{2}$. This is a consequence of the fact that
\begin{equation}
\label{eigenvalues}
\lambda_k = 4\pi^2 |k|^2 + 1.
\end{equation}
We now state the assumption on the interaction potential in our model.
\begin{assumption}
\label{interaction_potential_assumption}
Let $v:\Lambda \to \R$ be a positive even function such that $v \in L^\infty(\Lambda)$.
\end{assumption}
Finally, we fix a cut-off function $\chi$ with the following properties.
{\begin{assumption}
\label{cut-off_assumption}
Suppose that $\chi : \R \to [0,1]$ is an indicator function with
\begin{equation*}
\text{supp}(\chi) = [0,K],
\end{equation*}
for $K \leq K_{\mathrm{max}}$.
\end{assumption}}
\begin{remark}
The precise value of $K_{\mathrm{max}}$ (for the quintic local NLS) is determined by Proposition \ref{local_normalisability_lemma}. The existence of such an optimal constant was proved in \cite{OST22}, and this is the reason we focus on the proof for the quintic NLS.
\end{remark}
\subsection{General notation and conventions}
Throughout the paper, we will write $\N_0 = \{0,1,2,\ldots\}$ and $\mathbb{N} :=\{1,2,\ldots\}$. We will use $C$ to denote a generic positive constant that may change line-to-line. Moreover, if $C$ depends on the variables $a_1,\ldots,a_p$, we will write $C=C(a_1,\ldots,a_p)$. If $x \leq Cy$, we will sometimes write $x \lesssim y$, and if $x \lesssim C(a_1,\ldots,a_p) y$, we write $x \lesssim_{a_1,\ldots,a_p} y$. We will denote the indicator function of a set $A$ by $\chi_A$, where one defines
\begin{equation*}
\chi_A(x) = \begin{cases}
    1 \quad \text{ if } x \in A \\
    0 \quad \text{ if } x \notin A.
\end{cases}
\end{equation*}
We denote by $H^s(\Lambda)$ the $L^2$ based Sobolev space, and define
\begin{equation*}
\|f\|^2_{H^s(\Lambda)} := \sum_{k \in \Z} (1 +|k|^2)^{s} |\hat{f}(k)|^2.
\end{equation*}
Here we adopt the convention that
\begin{equation*}
\hat{f}(k) = \int \mathrm{d}x \, f(x) \, \e^{-2\pi \ii k x}.
\end{equation*}
We denote by $\mfhp$ the subset of $L^2(\Lambda^p)$ which satisfies
\begin{equation*}
\{f \in L^2(\Lambda^p) : f(x_1,\ldots,x_p) = f(x_{\pi 1}, \ldots, x_{\pi p})\}
\end{equation*}
for permutation $\pi \in S_p$. We identify closed linear operators on $\mfhp$ with their corresponding Schwartz kernel, which we denote $\xi(x_1,\ldots,x_p;y_1,\ldots,y_p) \equiv \xi(\mbx;\mby)$.

Let $\mc{H}$ be a separable Hilbert space. We denote the identity operator on $\mc{H}$ by $I$, and the set of bounded linear operators on $\mc{H}$ by $\mc{L}(\mc{H})$. For $p \in [1,\infty]$, we define
\begin{equation*}
\|\mc{B}\|_{\mc{L}^p(\mc{H})} := \begin{cases}
(\mathrm{Tr}_{\mc{H}} |\mc{B}|^p)^{1/p} \quad \text{ if } p< \infty \\
\sup \mathrm{spec} |\mathcal{B}| \quad \, \, \, \,\text{ if } p = \infty.
\end{cases}
\end{equation*}
Here $|\mc{B}| := \sqrt{B^*B}$ and we denote by $\mathrm{spec}$ the spectrum of an operator. Then we denote by the Schatten space $\mc{L}^p(\mc{H})$ the set of operators for which the norm $\|\cdot\|_{\mc{L}^p(\mc{H})}$ is finite. We will also denote by $\mc{L}^\infty(\mc{H}_1;\mc{H}_2)$ the space of compact operators from $\mc{H}_1$ to $\mc{H}_2$. Where clear, we will omit the spaces from norms.

\subsection{Classical setting}
\subsubsection{Quintic Hartree equation}
We now state the classical model that we consider. We consider the following nonlocal quintic NLS
\begin{align}
\nonumber
\ii \partial_t \vph + (\Delta -1 )\vph = &-\frac{2}{3} \int_{\Lambda^2} \dd y \,  \mathrm{d} z \,  v(x-y) v(y-z) |\vph(y)|^2 |\vph(z)|^2 \vph(x) \\
\label{quintic_Hartree}
& -\frac{1}{3} \int_{\Lambda^2} \dd y \,  \mathrm{d} z \,  v(x-y) v(x-z) |\vph(y)|^2 |\vph(z)|^2 \vph(x),
\end{align}
which is called the {\it (quintic) Hartree equation}. We note that by formally taking $v = \delta$ in \eqref{quintic_Hartree}, one recovers the local focusing quintic NLS. We adopt the same Poisson structure as in \cite[Section 1]{RS23}, which means that the Hamiltonian associated to \eqref{quintic_Hartree} is given by
\begin{equation*}
H(\vph) = \int_\Lambda \mathrm{d}x \, |\nabla \vph(x)|^2 + |\vph(x)|^2 
- \frac{1}{3} \int_{\Lambda^3} \dd x \, \dd y \, \dd z \, v(x-y)v(x-z) |\vph(x)|^2 |\vph(y)|^2 |\vph(z)|^2. 
\end{equation*}
\begin{remark}
The local well-posedness of \eqref{quintic_Hartree} for $v=\delta$ was proved by Bourgain in \cite{Bou93}. The corresponding local well-posedness for \eqref{quintic_Hartree} for $v \in L^1$ was proved in \cite[Section 2]{RS23}.
\end{remark}
\subsubsection{Gibbs measure and classical Hamiltonian}
Take $k \in \N$, and let $\mu_k$ be independent standard complex Gaussian measures defined by $\mu_k := \frac{1}{\pi}\e^{-|z|^2} \mathrm{d}z$, where $\dd z$ is the Lebesgue measure on  $\C$. We consider the measure
\begin{equation}
\label{Wiener_measure_defintion}
\mu_0 := \bigotimes_{k \in \N} \mu_k
\end{equation}
on the product space $(\C^\N, \mathcal{G},\mu_0)$. We call $\mu_0$ either the {\it (Gaussian) free field} or the {\it Wiener measure}. We denote elements of $\C^\N$ as $\om \equiv (\om_k)$. The {\it classical free field} is defined via
\begin{equation}
\label{element_free_field}
\vph := \sum_{k \in \N} \frac{\om_k}{\sqrt{\lambda_k}} u_k.
\end{equation}
Since $\mathrm{Tr}(h^{-1}) < \infty$, it follows that $\vph \in H^{\frac{1}{2}-\eta}$-$\mu_0$ almost surely for any $\eta > 0$. We recall that one can also view $\mu_0$ as a probability measure on $H^s(\Lambda)$ by considering a suitable pushforward; we direct the reader to \cite[Remark 1.3]{FKSS17}. We note that $\mu_0$ is the Gaussian measure on $H^s$ for $s< 1/2$ whose covariance is given by $h^{-1}$. 

We define the {\it mass} of the free field \eqref{element_free_field} by
\begin{equation*}
\mcN \equiv \mcN^\om := \|\varphi\|_{L^2(\Lambda)}^2.
\end{equation*}
We note that this is finite almost surely. We also define the {\it classical interaction} as
\begin{equation}
\label{classical_interaction}
\mc{V} \equiv \mc{V}^\om := \frac{1}{3}\int \dd x \, (v*|\vph|^2)^2(x)|\vph(x)|^2 = \frac{1}{3}\int \dd x \, \mathrm{d} y \, \mathrm{d} z \, v(x-y)v(x-z) |\varphi(x)|^2 |\varphi(y)|^2 |\varphi(z)|^2,
\end{equation}
where we recall that $v$ satisfies Assumption \ref{interaction_potential_assumption}. Again, we note that the Sobolev embedding theorem implies that this is finite almost surely. We define the {\it classical free Hamiltonian}
\begin{equation}
\label{classical_free_Hamiltonian}
H_0 := \int \dd x \, \mathrm{d} y \, h(x;y) \overline{\vph(x)}\vph(y).
\end{equation}
The {\it classical Hamiltonian} is then given by
\begin{equation*}
H \equiv H^\om := H_0 - \mc{V}
\end{equation*}
If $\xi$ is a closed densely-defined operator on $\mfhp$, we consider the random variable defined by
\begin{equation}
\label{classical_random_variable_definition}
\Theta(\xi) \equiv \Theta^\om(\xi) := \int \dd x_1 \ldots \mathrm{d} x_p \, \mathrm{d} y_1 \ldots \mathrm{d} y_p \, \xi(\mbx,\mby) \prod_{i=1}^p\overline{\vph(x_i)} \prod_{i=1}^p\vph(y_i).
\end{equation}
We note that by considering the operator $V$ on $\mfh^{(3)}$ which acts as multiplication by 
\begin{equation}
\label{three_body_multiplication}
\frac{1}{3}\left(v(x_1 - x_2)v(x_2-x_3) + v(x_2 - x_3)v(x_1-x_3) + v(x_1 - x_2)v(x_1-x_3)\right),
\end{equation}
one can write
\begin{equation}
\label{classical_Hamiltonian_rewrite}
H = \Theta(h) - \frac{1}{3}\Theta(V).
\end{equation}
We direct the reader to \cite[(1.37)--(1.39))]{RS23} for more details on this computation. Recalling the cut-off function from Assumption \ref{cut-off_assumption}, we can now define the {\it Gibbs measure} as
\begin{equation}
\label{focusing_Gibbs_rigorous}
\dd \mu \equiv \dd \mu_{\mathrm{Gibbs}}^\chi := \frac{1}{z^\chi_{\mathrm{Gibbs}}} \e^{\mc{V}} \chi(\mcN) \, \dd \mu_0,
\end{equation}
where we have taken the {\it classical partition function} to be the normalisation constant
\begin{equation}
\label{classical_partition_function}
z \equiv z^\chi_{\mathrm{Gibbs}} := \int \e^{\mc{V}} \chi(\mcN) \, \dd \mu_0.
\end{equation}
\begin{remark}
Let us remark that the probability measure $\mu$ and the partition function are both well-defined by Proposition \ref{local_normalisability_lemma} and Lemma \ref{nonlocal_normalisability_lemma} for the local and nonlocal (with bounded interaction potential) respectively. In the case that we take $\|v\|_{L^1} = 1$, Lemma \ref{uniform_partition_function_corollary} ensures that we can take the constant $K$ in Assumption \ref{cut-off_assumption} to be the cut-off for the local measure even when considering the nonlocal $\mc{V}$ with $v$ as in Assumption \ref{interaction_potential_assumption}. We direct the reader to \cite{Bou94} for the construction of the measure in the local case, and to \cite[Section 2]{RS23} for the nonlocal case. This ensures we have a rigorous construction of the measure in \eqref{focusing_Gibbs_formal}.
\end{remark}
For a random variable $X = X^\om$, we define the {\it classical (Gibbs) state} $\rho \equiv \rho^\chi$ by
\begin{equation*}
\rho(X) := \mathbb{E}_\mu [X] = \frac{1}{z} \int X \e^{\mc{V}} \chi(\mcN) \, \dd \mu_0.
\end{equation*}
We can now define our main object in the study of convergence.
\begin{definition}
For $p \in \N_0$, the {\it classical $p$-particle correlation function} is the operator on $\mfhp$ whose kernel is equal to
\begin{equation*}
\gamma^p(x_1,\ldots,x_p;y_1,\ldots,y_p) \equiv \gamma^{p,\chi}(x_1,\ldots,x_p;y_1,\ldots,y_p) = \rho\left( \prod_{i=1}^p \overline{\vph(y_i)} \prod_{i=1}^p \vph(x_i) \right).
\end{equation*}
\end{definition}

\subsection{Quantum System}
We will consider {\it bosonic Fock space}
\begin{equation}
\mc{F} \equiv \mathcal{F}(\mfh) := \bigoplus_{p \geq 0} \mfhp. 
\end{equation}
For $f \in \mfh$, we define the {\it rescaled creation and annihilation operators} $\vph^*_\eps(f)$ and $\vph_\eps(f)$ through their action on the $p^{th}$ sector of Fock space. Let $\Psi = (\Psi^{(p)}(x_1,\ldots,x_p))_{p \in \N} \in \mathcal{F}$. Define
\begin{align*}
(\vph^*_\eps(f)\Psi)^{(p)}(x_1,\ldots,x_p) &:= \sqrt{\frac{\eps}{p}} \sum_{i=1}^p f(x_i) \Psi^{(p-1)}(x_1,\ldots,x_{i-1},x_{i+1},\ldots,x_p), \\
(\vph_\eps(f)\Psi)^{(p)}(x_1,\ldots,x_p) &:= \sqrt{\eps(p+1)} \int \dd x \,  f(x) \Psi^{(p+1)}(x, x_1,\ldots,x_p).
\end{align*}
The operators $\vph^*_\eps(f)$ and  $\vph_\eps(f)$ are adjoints of one another, and are both closed and densely-defined. They satisfy the following commutation relations
\begin{equation*}
[\vph_\eps(f),\vph^*_\eps(g)] = \eps \langle f,g \rangle, \quad [\vph_\eps(f),\vph_\eps(g)] = [\vph^*_\eps(f),\vph^*_\eps(g)] = 0.
\end{equation*}
We consider the scaled creation and annihilation operators to be operator-valued distributions, and we denote their kernels by $\vph_\eps^*(x)$ and $\vph_\eps(x)$. One can compute that formally the kernels satisfy the commutation relations given by
\begin{equation*}
[\vph_\eps(x),\vph^*_\eps(\tilde{x})] = \eps \delta(x-\tilde{x}), \quad [\vph_\eps(x),\vph_\eps(\tilde{x})] = [\vph^*_\eps(x),\vph^*_\eps(\tilde{x})] = 0,
\end{equation*}
so the quantum fields formally commute in the semiclassical limit $\eps \to 0$.

For $h$ as in Assumption \ref{free_Hamiltonian_assumption}, one defines the {\it free quantum Hamiltonian} via
\begin{equation}
\label{quintic_quantum_free_Hamiltonian}
H_{\eps,0} := \int \dd x \,  \mathrm{d}y \, h(x;y) \vph_\eps^*(x) \vphe(y).
\end{equation}
For $v$ as in Assumption \ref{interaction_potential_assumption}, we also define the {\it quantum interaction} as
\begin{equation*}
\mc{V}_\eps := \frac{1}{3}\int \dd x \, \mathrm{d} y \, \mathrm{d} z \, v(x-y) v(x-z) \vph_\eps^*(x) \varphi^*_\varepsilon(y) \varphi^*_\varepsilon(z) \varphi_\varepsilon(x) \varphi_\varepsilon(y) \varphi_\varepsilon(z).
\end{equation*}
Then the {\it quantum Hamiltonian} is given by
\begin{equation*}
H_\eps := H_{\eps,0} - \mc{V}_\eps.
\end{equation*}
Moreover, we note that for 
\begin{equation}
\label{n-body_Hamiltonian}
H_\eps^p = \eps \sum_{i=1}^p (-\Delta_i + 1) - \frac{\eps^3}{3} \sum_{\substack{i,j,k \\ i \neq j \neq k \neq i}}  v(x_i-x_j)v(x_i-x_k),
\end{equation}
one can write
\begin{equation*}
H_\eps = \bigoplus_{p=0}^{\infty} H_\eps^{p}.
\end{equation*}
Analogously to the classical case \eqref{classical_random_variable_definition}, one can lift an operator $\xi$ on $\mfhp$ to Fock space via the formula
\begin{equation*}
\Theta_\eps(\xi) := \int \mathrm{d}x_1 \ldots \mathrm{d}x_p \, \mathrm{d}y_1 \ldots \mathrm{d}y_p \,
 \xi(\mbx,\mby) \prod_{i=1}^p \vphes(x_i) \prod_{i=1}^p \vphe(y_i).
 \end{equation*}
Then as in \eqref{classical_Hamiltonian_rewrite}, one can write
\begin{equation*}
H_\eps = \Theta_\eps(h) - \frac{1}{3}\Theta_\eps(V),
\end{equation*}
where we recall the operator $V$ defined by \eqref{three_body_multiplication}. We remark that this quantum set-up has been used in other papers addressing three-body interactions, see for example \cite{TV26}, which treats the derivation of the (defocusing) dynamics, and the references therein.
 \begin{remark}
Let us note the the quantisation used in this paper can be rephrased in terms of Wick quantisation; see Section \ref{Sec:Wick_quantisation}. This will be used when utilising the theory of Wigner measures developed in \cite{AN08}.
\end{remark}

The {\it rescaled particle number} is given by
\begin{equation*}
\mcN_\eps := \int \mathrm{d}x \, \vphes(x)\vphe(x).
\end{equation*}
Recalling the cut-off function from Assumption \ref{cut-off_assumption}, we define the {\it (truncated) quantum Gibbs state} as
\begin{equation*}
\rho_\eps(\mc{A}) \equiv \rho_{\eps}^\chi(\mc{A}) := \frac{1}{Z_{\eps}^\chi} \mathrm{Tr}(\mc{A} \, \e^{-H_{\eps}}\chi(\mcN_\eps)),
\end{equation*}
for a Fock space operator $\mc{A}: \mc{F} \to \mathcal{F}$. Here we define the {\it quantum partition function} to be the normalisation constant
\begin{equation*}
Z_{\eps} \equiv Z_{\eps}^\chi := \mathrm{Tr}(\e^{-H_{\eps}}\chi(\mcN_\eps)).
\end{equation*}
At times, it will also be convenient to consider the {\it (untruncated) free Gibbs state} given by
\begin{equation*}
\rho_{\eps,0}(\mc{A}) := \frac{1}{Z_{\eps,0}}\mathrm{Tr}(\mc{A}\e^{-H_{\eps,0}}),
\end{equation*}
where we have introduced the {\it free quantum partition function}
\begin{equation*}
Z_{\eps,0} := \mathrm{Tr}(\e^{-H_{\eps,0}}).
\end{equation*}
We will also consider the {\it quantum relative partition function} defined as
\begin{equation}
\label{relative_qunatum_partition}
\mc{Z}_\eps \equiv \mc{Z}_\eps^\chi := \frac{Z_\eps}{Z_{\eps,0}}.
\end{equation}
\begin{remark}
We note that $\rho_{\varepsilon,0}(\mathcal{N}_\varepsilon) \sim 1$, so the expected { (non-scaled)} number of particles in the free untruncated Gibbs state grows like $\varepsilon^{-1}$. So taking $\varepsilon^2 \sim \lambda$, one can rewrite \eqref{n-body_Hamiltonian} as
\begin{equation*}
H^{(p)}_\eps = \eps\left(\sum_{i=1}^p (-\Delta + 1) - \frac{\lambda}{3} \sum_{\substack{i,j,k \\ i \neq j \neq k \neq i}}^p  v(x_i-x_j)v(x_i-x_k)\right).
\end{equation*}
In particular, the $\eps$ parameter makes both of the terms inside the bracket of order $p$. So one can interpret the semiclassical limit $\eps \to 0$ as a mean-field limit. We direct the reader to the introduction of \cite{FKSS17} for further details.
\end{remark}
We can now introduce the quantum analogue of the $p$-particle correlation functions.
\begin{definition}
For $p \in \N_0$, the {\it quantum $p$-particle correlation function} is the operator on $\mfhp$ whose kernel is given by
\begin{equation*}
\gamma^{p}_\eps(x_1,\ldots,x_p;y_1,\ldots,y_p) \equiv \gamma^{p,\chi}_\eps(x_1,\ldots,x_p;y_1,\ldots,y_p) := \rho_{\eps}\left(\prod_{i=1}^p \vphes(y_i) \prod_{i=1}^p \vphe(x_i)\right)
 \end{equation*}
\end{definition}
\subsection{Main Results}
Our first result is for interaction potentials which satisfy Assumption \ref{interaction_potential_assumption}.

\begin{theorem}[Bounded interaction potentials]
\label{bounded_theorem}
Let $v$ be as in Assumption \ref{interaction_potential_assumption} and suppose that $\chi$ is as in Assumption \ref{cut-off_assumption}. Let $p \in \N$. Then we have
\begin{equation*}
\lim_{\eps \to 0} \|\gamma_\eps^p -\gamma^p\|_{\mc{L}^1(\mfhp)} = 0.
\end{equation*}
Moreover, recalling the definitions of $\mc{Z}_\eps$ and $z$ from \eqref{relative_qunatum_partition} and \eqref{classical_partition_function}, one has
\begin{equation*}
\lim_{\eps \to 0} \mc{Z}_\eps = z.
\end{equation*}
\end{theorem}
{\begin{remark}
    Let us remark that we can relax the assumption that $K \leq K_{\mathrm{max}}$ in Theorem \ref{bounded_theorem}. Indeed, we note that by Lemma \ref{nonlocal_normalisability_lemma} below, there is no optimal cut-off when $v \in L^\infty(\T)$. We state Theorem \ref{bounded_theorem} in this form because we will use it in this form in the proof of the $\delta$ potential.
\end{remark}}

By taking an approximation of the identity, we are able to prove a result for the local interaction, corresponding to $v = \delta$. We take $U : \R \to \R$ to be a continuous even function with {$U \geq 0$ and}
\begin{equation*}
\int \dd x \, U(x) = 1.
\end{equation*}
Define $v^N:\T \to \R$ by
\begin{equation}
\label{approximation_identity}
v^N(x) := N U(N[x]),
\end{equation}
where $[x]$ is the unique element of $(x+\Z)\cap \T$. Then $v^N$ satisfies Assumption \ref{interaction_potential_assumption} and converges weakly to $\delta$ with respect to continuous functions as $N \to \infty$. We denote any object $X$ like the partition function or the correlation functions associated to $v^N$ by $X^N$.
\begin{theorem}[Convergence for delta potential]
\label{delta_theorem}
Let $v= \delta$ and $v^N$ be as in \eqref{approximation_identity}. Suppose that $\chi$ is as in Assumption \ref{cut-off_assumption}. Then there is a sequence $N_\eps \to \infty$ as $\eps \to 0$ such that for any $p \in\N$, one has
\begin{equation*}
\lim_{\eps \to 0} \|\gamma_{\eps}^{p,N_\eps} -\gamma^p\|_{\mc{L}^1(\mfhp)} = 0.
\end{equation*}
Moreover, one has
\begin{equation*}
\lim_{\eps \to 0} \mc{Z}_\eps^{N_\eps} = z.
\end{equation*}
\end{theorem}
\begin{remark}
We make a number of remarks about Theorems \ref{bounded_theorem} and \ref{delta_theorem}.
\begin{itemize}
\item Let us emphasise that the main novelty in Theorems \ref{bounded_theorem} and \ref{delta_theorem} is the ability to treat cut-offs with no smoothness assumptions, in particular the indicator function $\chi_{[0,K]}$.
\item  We expect our results to hold for any superharmonic trap, $V(x) = |x|^{\sigma}$ for $\sigma > 2$, when $\Lambda := \R$ and for sufficiently small cut-offs. We do not comment further on this problem. 
\item Our method also extends easily to the case of the cubic NLS, allowing us also to extend the result of \cite{RS22} to rough cut-offs.
\item In \cite{RS22,RS23} theorems were given in both the time independent and time dependent settings. Our results readily extend to the time dependent setting on the torus, since the arguments do not depend on the smoothness of the cut-off.
\item A similar result to Theorem \ref{delta_theorem} for integrable interaction potentials holds under the assumption that the cut-off has sufficiently small support that the corresponding Gibbs measure is well defined. We direct the reader to \cite[Theorem 1.9]{RS23} for a precise statement of the corresponding results for smooth cut-offs.
\end{itemize}
\end{remark}
\begin{remark}
\label{LNZ_remark}
Our result is similar to a recent result by L\"u, Nam, and Zhu in \cite{LNZ26}, where the authors also provide a derivation of the focusing $\Phi_1^6$ measure with optimal cut-off. In their result, they are able to extract a rate for the $N$ limit in terms of the parameter $\eps$. In their result, they start with a smooth truncation in the quantum problem, see \cite[Assumption 2.2]{LNZ26}. We are able to directly consider the rough cut-off in the quantum model. Unlike \cite{LNZ26}, which uses a variational approach to the problem, our approach is based on the perturbative expansion introduced in \cite{FKSS17}.
\end{remark}

\subsection{Future directions}
An interesting problem, as raised in \cite{DR24}, is to consider the problem when mass renormalisation is required. We note that in dimensions higher than one, there is a no-go theorem in \cite{BS96}, so one would need to consider the case of the nonlocal NLS --- see \cite{Bou97,OTT24} for constructions of the measure. It is possible to study the local problem with mass renormalisation if one considers a trapping potential which is subharmonic. However we are unable to treat these problems since the bounds on the remainder term in \cite{RS22,RS23} are based on the Feynman-Kac formula, and cannot be easily adapted to settings where renormalisation is required.

It would also be interesting to see if the methods introduced in this setting can be used to study the corresponding canonical ensemble problem, also studied in \cite{DR24}. The primary difficulty in this problem is the lack of Wick-type theorems, which we use to prove some estimates satisfied by the free quantum state. We plan to revisit this problem in future work.

\subsection{Previously Known Results}
The problem of deriving Gibbs measures from many-body quantum mechanics has received a lot of attention in recent years. In the defocusing case, the problem first was studied by Lewin--Nam--Rougerie and Fr\"ohlich--Knowles--Schlein--Sohinger. We direct the reader to \cite{CKRTG26,FKSS17,FKSS20,FKSS22,FKSS23,FKSS22b,LNR15,LNR18,LNR21,NZZ25,NZZ26,Soh19} for results in the defocusing case. Further results were proved in \cite{GSS26,JR25}. We direct the reader to \cite{RS23} for a more in depth discussion of the methods used. We also mention \cite{FKSS18}, which treats the dynamical problem in one dimension, giving a derivation of the time dependent correlation functions. Moreover, the first author with Ammari and Petrat gave a higher-order expansion in the case of a finite graph in \cite{AFP25}. 

The focusing problem was studied using a perturbative expansion by the second author and Sohinger in \cite{RS22,RS23}, and derivations of the Gibbs measure associated to the focusing cubic and quintic NLS were given in the respective papers. In these papers, the proof of the convergence of the explicit terms was based on an application of the Helffer-Sj\"ostrand formula. This required a smoothness assumption, as well as a bound on the corresponding non-truncated quantum explicit terms. They also gave a derivation of the time dependent correlation functions. L\"u, Nam, and Zhu also recently gave a derivation of the focusing $\Phi_1^6$ with optimal cut-off in \cite{LNZ26}. This was done using rather different methods to those in \cite{RS22,RS23}, being based on the quantum de Finetti theorem, Berezin-Lieb inequality, and a variational principle. In \cite{LNZ26}, the authors extract a rate for the $N$ limit in terms of the parameter $\eps$. They also studied the problem of the relative partition function when the support is larger than the optimal cut-off, showing that it displays blow-up in the semiclassical parameter; see \cite[Theorem 2.10]{LNZ26}. The variational method was also used by Dinh and Rougerie in \cite{DR24}, where the authors gave a derivation of the canonical Gibbs measure.

\subsection{Method of Proof}
The proof is based on the perturbative expansion introduced in \cite{FKSS17}. For bounded potentials, one gets analytic power series via the same arguments as \cite{RS22,RS23}. Thus the primary difficulty is the convergence of the explicit terms, since we cannot use the Helffer-Sj\"ostrand formula approach used previously. To get around this, we use the theory of Wigner measures introduced in \cite{AN08}. First, we use a result from \cite{AN11} to deduce that the Wigner measure associated to the truncated free Gibbs state is the truncated free field. We then use commutator estimates and an inductive argument to reduce convergence of the explicit terms to the convergence of the truncated free Gibbs state.

One of the difficulties we face is that the convergence given by Wigner measures does not automatically give (uniform) convergence in the space of observables we need, nor for the class of potentials we consider. To get around this, we use the fact that the density matrices for the quantum explicit terms are able to uniformly absorb arbitrarily small powers of the free Hamiltonian. We can trade this extra regularity to ensure convergence.

The new convergence argument also means we do not need bounds on the untruncated quantum explicit terms, in contrast to \cite{RS22,RS23}. In particular, this means that we avoid needing to appeal to the diagrammatic arguments used in \cite{FKSS17}. To the authors' knowledge, this is the first perturbative expansion derivation to completely avoid the diagrammatic arguments.

\subsection{Outline of Paper}
In Section \ref{preliminary_section} we outline some background results we will need about the Gibbs measure, and give a brief introduction to the theory of Wigner measures. Section \ref{SEC:semiclassical_analysis} is devoted to the proof of a number of results from semiclassical analysis. Section \ref{SEC:bounded} deals with the problem of bounded interaction potentials, giving a proof of Theorem \ref{bounded_theorem}. In Section \ref{Section:unbounded_potentials}, we extend these results to the local interaction, and give the proof of Theorem \ref{delta_theorem}. Finally, in Appendix \ref{estimates_appendix}, we prove a uniform bound for the free quantum Gibbs state.

\section{Preliminary Results}
\label{preliminary_section}
\subsection{Classical bounds}
We recall \cite[Lemma 3.10]{Bou94} (see also \cite{OST22}), which concerns the normalisability of the Gibbs measure for the local NLS.
\begin{proposition}
\label{local_normalisability_lemma}
Let $F : \R \to [0,1]$ be a function whose support is contained in $[0,K]$. Then for $p \in [2,3)$ and any $K > 0$, one has
\begin{equation*}
\e^{\frac{1}{p} \|\varphi\|_{L^{2p}(\T)}^{2p}}F(\|\varphi\|_{L^2}^2) \in L^1(\dd \mu_0).
\end{equation*}
Moreover, for $K \leq K_{\mathrm{max}}$, one has
\begin{equation*}
\e^{\frac{1}{3} \|\varphi\|_{L^6(\T)}^6}F(\|\varphi\|_{L^2}^2) \in L^1(\dd \mu_0).
\end{equation*}
\end{proposition}
\begin{remark}
{The normalisability up to the optimal value $K_{\mathrm{max}}$ was determined for the quintic NLS in \cite{OST22}.}
\end{remark}
In what follows, it will be convenient to consider the Hartree equation as an approximation to the local NLS. We need the following result, which ensures that the Gibbs measure is well defined for bounded interaction potentials.
\begin{lemma}
\label{nonlocal_normalisability_lemma}
Suppose that $\chi$ is as in the statement of Proposition \ref{local_normalisability_lemma} with any choice of $K < \infty$. Suppose the interaction potential $v \in L^\infty(\T)$. Then one has
\begin{equation*}
\e^{\mcV(\varphi)} \chi(\|\varphi\|_{L^2}^2) \in L^1(\dd \mu_0).
\end{equation*}
\end{lemma}
\begin{proof}
One has
\begin{equation*}
\left|\mcV(\varphi)\right| \leq \frac{1}{3}\|v\|^2_{L^\infty}\|\varphi\|_{L^{2}}^{6},
\end{equation*}
which follows from applying H\"older's inequality to \eqref{classical_interaction}. The result then immediately follows since the weight is bounded.
\end{proof}
We also need the following result, which allows us to pick the cut-off uniformly when considering the interaction potential $v^N$ as in \eqref{approximation_identity}.
\begin{lemma}
\label{uniform_partition_function_corollary}
Suppose that $v^N$ is as in \eqref{approximation_identity}. Suppose that $\chi$ is as in the statement of Proposition \ref{local_normalisability_lemma} with $K$ sufficiently small such that the local measure corresponding to $v=\delta$ and $p=6$ is normalisable. Then
\begin{equation*}
\e^{\mcV^N(\varphi)} \chi(\|\varphi\|_{L^2}^2) \in L^1(\dd \mu_0).
\end{equation*}
Moreover, recalling the definition of the classical partition function in \eqref{classical_partition_function}, one has that $z^N$ is uniformly bounded in $N$.
\end{lemma}
\begin{proof}
Applying Young's inequality to the \eqref{classical_interaction}, one has
\begin{equation*}
\left|\mcV^N(\varphi)\right| \leq \frac{1}{3}\|v^N\|^2_{L^1}\|\varphi\|_{L^{6}}^{6} = \frac{1}{3}\|\varphi\|_{L^{6}}^{6},
\end{equation*}
The result then follows from Proposition \ref{local_normalisability_lemma}.
\end{proof}
We also recall the following bound, see for example \cite[Lemma 3.2]{RS22}.
\begin{lemma}
\label{random_variable_bound_lemma}
Suppose that $\xi \in \mc{L}(\mfhp)$. Then
\begin{equation*}
|\Theta(\xi)| \leq \|\xi\| \|\vph\|_{\mfh}^{2p}.
\end{equation*}
\end{lemma}
\subsection{Wick quantisation and Weyl operators}
\label{Sec:Wick_quantisation}
We need to briefly recall the Wick calculus and Weyl operators. We direct the reader to \cite{AN08} for a more detailed discussion.
\begin{definition}[Homogeneous Polynomials]
We say that an operator $b \in \mathcal{P}_{p,q}$ if there is $\tilde{b} \in \mathcal{L}(\mfh^{(p)},\mfh^{(q)})$ such that $b(\varphi) = \langle \varphi^{\otimes_q},\tilde{b} \varphi^{\otimes_p} \rangle$. We define
\begin{equation*}
|b|_{\mathcal{P}_{p,q}} := |\tilde{b}|_{\mathcal{L}(\mfh^{(p)},\mfh^{(q)})}.
\end{equation*}
\end{definition}
We also denote by $ \mathcal{P}^\infty_{p,q}$ the set of compact polynomials where $\tilde{b} \in \mathcal{L}^\infty(\mfh^{(p)},\mfh^{(q)})$. 
\begin{remark}
Let use remark that $b(\varphi)  = \|\varphi\|_\mfh^2$ is in $\mc{P}_{1,1}$, with $\tilde{b}$ equal to the identity.
\end{remark}
\begin{definition}
The Wick quantisation of $b \in \mc{P}_{p,q}$ is the operator $b_\eps^{\mathrm{W}}$ given by
{ 
  \begin{align*}
        b_\eps^{\mathrm{W}}|_{\mathfrak{h}^{(n)}}:= 1_{[p,+\infty)} (n)  \frac{\sqrt{n!(n+q-p)!}}{(n-p)!}  \eps^{\frac{p+q}{2}}  \mathcal{S}_{n-p+q} (\tilde b \otimes 1^{\otimes (n-p)}),
    \end{align*}
    or equivalently using the kernel representation
}
\begin{equation*}
b_{\eps}^{\mathrm{W}}  \equiv b^{\mathrm{W}} =\Theta_{\eps}(\tilde{b}) = \int \dd x_1 \ldots \mathrm{d} x_p \, \dd y_1 \ldots \mathrm{d} y_q \, \tilde{b}(x_1,\ldots,x_p;y_1\ldots,y_q) \prod_{i=1}^p \varphi^*_\eps(x_i) \prod_{i=1}^q \varphi_\eps(y_i).
\end{equation*}
\end{definition}
\begin{remark}
Moreover, since $v \in L^\infty$, we have that $\mc{V}_{\eps} = (\Theta(V))^\WW$ and $\mc{V} \in \mc{P}_{3,3}$. 
\end{remark}
We will also use Weyl operators in our proof. We make the following definition.
\begin{definition}
For $ f \in \mfh$, we define the Weyl operator $W_\eps(f)$ as
\begin{equation*}
W_\eps(f) := \e^{\frac{\ii}{\sqrt{2}}(\vph_\eps(f) + \vph^*_\eps(f))}
\end{equation*}
\end{definition}
We have the following commutation relations for the Weyl operators. Let $f,g \in \mfh$.
\begin{align}
\nonumber
W^*_\eps(g) W_\eps(f) W_\eps(g) &= \e^{-\ii \eps \mathrm{Im}\langle f,g\rangle} W_\eps(f), \\
\label{Weyl_CR}
W_\eps(g)W_\eps(f) &= \e^{-\frac{\ii \eps}{2}\mathrm{Im}\langle g,f\rangle} W_\eps(g+f).
\end{align}
\subsection{Wigner measures}
In this section, we briefly recall the notion of a {\it Wigner measure} and prove a preliminary lemma. We direct the reader to \cite{AN08} for a detailed introduction to the theory of Wigner measures.
\begin{definition}
\label{Wigner_definition}
A Borel probability measure $\nu$ is a Wigner measure for a family of trace class operators $(\varrho_\eps)_{\eps \in (0,1)}$ on a Hilbert space $\mathcal{H}$ if there is a countable sequence $\eps_n \to 0$ such that for any $f \in \mc{H}$ one has
\begin{equation*}
\lim_{\eps_n \to 0} \mathrm{Tr}(W_{\eps_n}(f)\varrho_{\eps_n}) = \int_{\mc{H}} \e^{\ii \sqrt{2} \mathrm{Re}\langle f,\varphi \rangle} \dd \nu(\varphi)
\end{equation*}
We denote the set of Wigner measures admitted by $(\varrho_\eps)_{\eps \in (0,1)}$ by $\mc{M}(\varrho_\eps)$.
\end{definition}
\begin{remark}
If $\nu$ is a Wigner measure for the family $(\varrho_\eps)_{\eps \in (0,1)}$ of trace-class operators which can absorb arbitrary powers of the number operator, one has that
\begin{equation*}
\lim_{\eps_n \to 0} \mathrm{Tr}( b^{\WW}_{\eps_n} \varrho_{\eps_n})= \int_{\mc{H}} b \, \dd \nu
\end{equation*}
for all $b \in \mathcal{P}^\infty_{p,q}$. 
\end{remark}
The following lemma gives a sufficient condition for the existence of a Wigner measure, see \cite{AN08}.
\begin{lemma}
\label{Wigner_existence_lemma}
Suppose that $(\varrho_\eps)_{\eps \in (0,1)}$ is  a family of density operators such that
\begin{equation*}
\mathrm{Tr}(\mcN_\eps^k\varrho_\eps) \leq C(k)  < \infty
\end{equation*}
for some $k > 0$ uniformly in $\eps \in (0,1)$. Then $\mc{M}(\varrho_\eps)$ is non-empty and there is a Wigner measure satisfying
\begin{equation*}
\int \|\varphi\|_{\mc{H}}^{2k} \, \dd \nu (\varphi) < \infty.
\end{equation*}
\end{lemma}
\begin{remark}
In fact, the assumption in Lemma \ref{Wigner_existence_lemma} that we have a sequence of quantum states can be relaxed to the assumption that we have a sequence of trace-class operators which satisfy
\begin{equation*}
\|(\mcN_\eps + 1)^\frac{k}{2} A_\eps (\mcN_\eps + 1)^\frac{k}{2}\|_{\mc{L}^1} < C(k) < \infty
\end{equation*}
uniformly in $\eps$. See \cite[Proposition 6.4]{AN08}.
\end{remark}
\begin{remark}
 Let us note that
\begin{equation*}
\frac{1}{Z_{\eps,0}}\mathrm{Tr}(\mcN_\eps^p\e^{-H_{\eps,0}}\chi(\mcN_\eps)) \leq K^p \frac{1}{Z_{\eps,0}}\mathrm{Tr}(\e^{-H_{\eps,0}}) = K^p.
\end{equation*}
Here we have used the positivity of the Gibbs state and Assumption \ref{cut-off_assumption}.
\end{remark}
\section{Semiclassical Analysis}
\label{SEC:semiclassical_analysis}
In this section, we collect some results and estimates from semiclassical analysis.
\subsection{Wigner measures for free states}
In this section, we consider the Wigner measures for the sequence of densities given by
\begin{equation*}
\rho_{\eps,0} = \frac{\e^{-H_{\eps,0}}}{Z_{\eps,0}}, \quad \rho^\chi_{\eps,0} = \frac{\e^{-H_{\eps,0}}}{Z_{\eps,0}}\chi(\mcN_\eps).
\end{equation*}
We have the following result.
\begin{lemma}
\label{Wigner_free_measure_lemma}
One has that
\begin{equation*}
\mathcal{M}((\rho_{\eps,0})_{\eps \in (0,1)}) = \{\dd \mu_0\}.
\end{equation*}
\end{lemma}
\begin{proof}
We recall \cite[Proposition 5.2.28]{BR81}, so one has that
\begin{equation}
\label{free_Gibbs_characteristic}
\frac{1}{Z_{\eps,0}} \mathrm{Tr}(W_\eps(f)\e^{-H_{\eps,0}})= \mathrm{exp}\left(-\frac{\eps}{4}\left\langle f, \frac{1 + \e^{-\eps h}}{1-\e^{-\eps h}} f\right\rangle\right) \longrightarrow \mathrm{exp}\left(-\frac{1}{2}\left\langle f,h^{-1} f\right\rangle\right),    
\end{equation}
which is the characteristic of the Gaussian free field.
\end{proof}
The following result is an important consequence from Lemma \ref{Wigner_free_measure_lemma}.
\begin{lemma}
\label{convergence_free_cutoff_measure_lemma}
One has that
\begin{equation*}
\mathcal{M}\left(\frac{1}{Z_{\eps,0}}\e^{-H_{\eps,0}} \chi(\mcN_{\eps})\right) = \left\{\chi(\mcN) \, \dd \mu_0 \right\}
\end{equation*}
\end{lemma}
\begin{proof}
We need to show that the quantum characteristic of $\rho_{\eps,0}^\chi$ converges to the characteristic of $\chi(\mcN) \dd \mu_0$. Note that by the Weyl commutation relations \eqref{Weyl_CR} and \eqref{free_Gibbs_characteristic}, it follows that the Wigner measure of $W_{\eps}(f) \rho_{\eps,0}$ is $\e^{\sqrt{2}\pi \ii \mathrm{Re} \langle f,\cdot\rangle} \, \dd \mu_0$. In particular, one can apply \cite[Corollary 3.8]{AN11} to obtain that
\begin{equation*}
\lim_{\eps \to 0} \mathrm{Tr}(W_\eps(f) \rho_{\eps,0}\chi(\mcN_\eps)) \to \int \e^{\sqrt{2} \pi \ii \mathrm{Re}\langle f,u\rangle} \chi(\mcN) \, \dd \mu_0.
\end{equation*}
Here we have used that the inverse Fourier transform of an indicator function is bounded and measurable.
\end{proof}
\subsection{Non-compact Wick symbols}
\label{Section:non-compactness}
In what follows, we will need to consider the convergence of operators which are the Wick quantisation of symbols in $\mathcal{L}_{p,p}$, rather than $\mcL^{\infty}_{p,p}$. Therefore, we cannot directly apply the results of \cite{AN08}. Instead, we will prove two results for states which can absorb semiclassical powers of the Laplacian. For simplicity of notation, we will slightly abuse notation and write
\begin{align*}
h^{(s)} &:= (-\Delta)^s + 1, \\
H_{\eps,0}^{(s)} &:= \dd \Gamma_\eps(h^{(s)}).
\end{align*}
We have the following bound.
\begin{lemma}
\label{Hamiltonian_bound}
Let $\xi \in \mc{L}(\mfh^{(p)})$ and take $s \in (0,1/2)$. The following uniform in $\eps$ bounds hold.
 \begin{align*}
\|(H_{\eps,0}^{(s)} + 1)^{-\frac{p}{2}}
\Theta_\eps(\xi)(H_{\eps,0}^{(s)} + 1)^{-\frac{p}{2}}\|_{\mc{L}} &\leq \left\Vert \left(\left(h^{(s)}\right)^{-\frac{1}{2}}\right)^{\otimes p} \xi\left(\left(h^{(s)}\right)^{-\frac{1}{2}}\right)^{\otimes p} \right \Vert_{\mathcal{L}(\mfh^{\otimes p})}, \\
\|(H_{\eps,0}^{(s)} + 1)^{-\frac{p}{2}}
\Theta_\eps(\xi) (\mcN_\eps + 1)^{-\frac{p}{2}}\|_{\mc{L}} &\leq \left\Vert \left( \left( h^{(s)} \right)^{-1/2}\right)^{\otimes p} \xi \right \Vert_{\mathcal{L}(\mfh^{\otimes p})}, \\
\| (\mcN_\eps + 1)^{-\frac{p}{2}} \Theta_\eps(\xi)(H_{\eps,0}^{(s)} + 1)^{-\frac{p}{2}}\|_{\mc{L}} &\leq \left\Vert \xi \left( \left( h^{(s)} \right)^{-1/2}\right)^{\otimes p} \right \Vert_{\mathcal{L}(\mfh^{\otimes p})}.
\end{align*}
\end{lemma}

{\begin{proof}
    Let $\Phi,\Psi \in \mathcal{D}$ where $\mathcal{D}$ is a dense domain. Then  
    \begin{align*}
         &\langle \Phi, \Theta_\eps (\xi)  \Psi \rangle  = \sum_{n\in \mathbb{N}} \langle \Phi^{(n)}, \Theta_\eps (\xi)  \Psi^{(n)} \rangle = \sum_{n\in \mathbb{N}} {p!}  \eps^p \binom{n}{p}\langle \Phi^{(n)}, \xi  \otimes 1^{(n-p)}\Psi^{(n)} \rangle 
          \\ &= \sum_{n\in \mathbb{N}}   \eps^p n(n-1)\cdots (n-p+1)\Big \langle \Phi^{(n)},  \left[  \left (\left(h^{(s)}\right)^{1/2} \right)^{\otimes p}  \otimes 1^{(n-p)} \right]   \left[  \left (\left(h^{(s)}\right)^{-1/2} \right)^{\otimes p}  \otimes 1^{(n-p)} \right]
          \\ & \qquad \qquad \times \left( \xi   \otimes 1^{(n-p)} \right)    \left[  \left (\left(h^{(s)}\right)^{-1/2} \right)^{\otimes p}  \otimes 1^{(n-p)} \right]  \left[  \left (\left(h^{(s)}\right)^{1/2} \right)^{\otimes p}  \otimes 1^{(n-p)} \right]\Psi^{(n)}\Big  \rangle
          \\ & \leq \sum_{n\in \mathbb{N}}   \eps^p n(n-1)\cdots (n-p+1) \left\Vert    \left (\left(h^{(s)}\right)^{1/2} \right)^{\otimes p}  \otimes 1^{(n-p)} \Phi^{(n)}\right \Vert_{\mfh^{(n)}} 
          \\ &\quad \times  \left\Vert    \left (\left(h^{(s)}\right)^{1/2} \right)^{\otimes p}  \otimes 1^{(n-p)}  \Psi^{(n)}\right \Vert_{\mfh^{(n)}}
            \left \Vert   \left[  \left( \left( h^{(s)}\right)^{-1/2}\right)^{\otimes p} \xi  \left( \left( h^{(s)}\right)^{-1/2}\right)^{\otimes p} \right]     \otimes 1^{(n-p)}  \right \Vert_{\mathcal{L} (\mathfrak{h}^{\otimes n} )}.
    \end{align*}
    Here we have used Cauchy-Schwarz in the above computation. Denote
    \begin{align}
        h^{(s)}_j:= 1\otimes \cdots  \otimes \underbrace{h^{(s)}}_{\text{jth position}} \otimes \cdots \otimes 1.\label{jthopertor}
    \end{align}
 We have
    \begin{align*}
     \eps^p n(n-1)\cdots (n-p+1) &\left\Vert  \left (\left(h^{(s)}\right)^{1/2} \right)^{\otimes p}  \otimes 1^{(n-p)}\Phi^{(n)}\right \Vert_{\mfh^{(n)}}^{ 2} \\
     & =     \eps^p n(n-1)\cdots (n-p+1) \left \langle \Phi^{(n)} ,    \left(h^{(s)}\right)^{\otimes p}  \otimes 1^{(n-p)}  \Phi^{(n)} \right \rangle 
    \\ &=   \eps^p n(n-1)\cdots (n-p+1) \left \langle \Phi^{(n)} ,   h^{(s)}_1 h^{(s)}_2 \cdots h^{(s)}_p   \Phi^{(n)} \right \rangle 
    \\ &\leq \left \langle \Phi^{(n)} , \sum_{ 1\leq i_1,i_2, \dots , i_p\leq n} \eps ^p  h^{(s)}_{i_1}  h^{(s)}_{i_2} \cdots  h^{(s)}_{i_p}   \Phi^{(n)} \right \rangle 
    \\ & \leq  \left \Vert \left( H_{\eps,0}^{(s)}\right)^{p/2}   \Phi^{(n)} \right \Vert^2_{\mfh^{(n)}} \leq \left \Vert  \left( H_{\eps,0}^{(s)}+1\right)^{p/2}   \Phi^{(n)} \right \Vert^2_{\mfh^{(n)}},
    \end{align*}
    where we have used  $ h^{(s)}\geq 0$ and $\Phi^{(n)}$ symmetric. Similarly 
    \begin{align*}
     \eps^p n(n-1)\cdots (n-p+1) \left\Vert  (h^{(s)})^{\otimes p}  \otimes 1^{(n-p)}  \Psi^{(n)}\right \Vert_{\mfh^{(n)}}^2
  \leq   \left \Vert  \left( H_{\eps,0}^{(s)}+1\right)^{p/2}   \Psi^{(n)} \right \Vert^2_{\mfh^{(n)}}.
    \end{align*}
   Using the above two estimates, we get 
   \begin{multline*}
    \langle \Phi, \Theta_\eps (\xi)  \Psi \rangle_{\mcF}
           \leq \sum_{n\in \mathbb{N}}  \left \Vert \left( H_{\eps,0}^{(s)}+1\right)^{p/2}  \Psi^{(n)} \right \Vert_{\mfh^{(n)}} \left \Vert \left( H_{\eps,0}^{(s)}+1\right)^{p/2}  \Phi^{(n)} \right \Vert_{\mfh^{(n)}} \\
        \times
          \left \Vert    \left( \left( h^{(s)}\right)^{-1/2}\right)^{\otimes p} \xi  \left( \left( h^{(s)}\right)^{-1/2}\right)^{\otimes p}      \right \Vert_{\mathcal{L} (\mathfrak{h}^{\otimes p} )}
          \\ \leq  \left \Vert \left( H_{\eps,0}^{(s)}+1\right)^{p/2}  \Psi \right \Vert_\mcF \left \Vert \left( H_{\eps,0}^{(s)}+1\right)^{p/2}  \Phi \right \Vert_\mcF
           \left \Vert   \left( \left( h^{(s)}\right)^{-1/2}\right)^{\otimes p} \xi  \left( \left( h^{(s)}\right)^{-1/2}\right)^{\otimes p}       \right \Vert_{\mathcal{L}(\mathfrak{h}^{\otimes p } )},
    \end{multline*}
    where
    \begin{align*}
        & \left \Vert   \left( \left( h^{(s)}\right)^{-1/2}\right)^{\otimes p} \xi  \left( \left( h^{(s)}\right)^{-1/2}\right)^{\otimes p}       \right \Vert_{\mathcal{L}(\mathfrak{h}^{\otimes p } )}
        \\&=\left\Vert ((-\Delta)^s+1)^{-1/2} \otimes \cdots \otimes ((-\Delta)^s+1)^{-1/2}  \ \ \xi   \ \ ((-\Delta)^s+1)^{-1/2} \otimes \cdots \otimes ((-\Delta)^s+1)^{-1/2} \right \Vert_{\mathcal{L}(\mathfrak{h}^{\otimes p } )}.
    \end{align*}
    The remaining two estimates follow similarly. We omit the details.
    \end{proof}}
We have the following corollary of Lemma \ref{Hamiltonian_bound}. 
\begin{corollary}
\label{quantum_tail_bound}
Suppose that $\xi \in \mc{L}^2$. Then
\begin{equation*}
\|(H_{\eps,0}^{(s)}+1)^{-\frac{p}{2}}\Theta_\eps(\xi)(H_{\eps,0}^{(s)}+1)^{-\frac{p}{2}}\|_{\mc{L}} \lesssim \|\xi\|_{H^{-s}(\T^{2p})}.
\end{equation*}
\end{corollary}
\begin{proof}
By Lemma \ref{Hamiltonian_bound}, we have
\begin{multline}
\label{messy_norm_bound}
\|(H_{\eps,0}^{(s)}+1)^{-\frac{p}{2}}\Theta_\eps(\xi)(H_{\eps,0}^{(s)}+1)^{-\frac{p}{2}}\|_{\mc{L}} \\
\lesssim \left\|\prod_{j=1}^p (1 + (-\Delta_{x_j})^s)^{-\frac{1}{2}} \int  \dd \mby \, \xi(\mbx;\mby) \prod_{j=1}^p (1 + (-\Delta_{y_j})^s)^{-\frac{1}{2}} \psi(\mby)\right\|_{L^2_{\mbx}(\T^p)}.
\end{multline}
We apply Cauchy-Schwarz in $\mby$ so it follows that \eqref{messy_norm_bound} is bounded by
\begin{equation*}
\|\xi\|_{H^{-s}_\mbx H^{-s}_{\mby}}\|\psi\|_{L^2_{\mby}},
\end{equation*}
from which the result follows. Here we have used the equivalence of norms induced by $(1-\Delta)^{-\frac{s}{2}}$ and $(1+ (-\Delta)^s)^{-\frac{1}{2}}$.
\end{proof} 
We can now prove the following important result.
{\begin{lemma}
\label{Wigner_measure_V_lemma}
Suppose that $\varrho_{\eps}$ is such that $\varrho_{\eps} = \varrho_{\eps}\chi(\mcN_\eps)$, $ [\varrho_\eps,\mathcal{N}_\eps]=0$, and that
\begin{equation}\label{asslemma}
\left\|\varrho_\eps\left(H_{\eps,0}^{(s)} +1 \right)^{\frac{3}{2}}\right\|_{\mcL^1} < C < \infty
\end{equation}
uniformly in $\eps$ for any $s \in (0,1/2)$. Moreover suppose that
\begin{equation*}
\mathcal{M}(\varrho_\eps) = \{\dd \mu\}.
\end{equation*}
Then for $v \in L^\infty(\T)$, one has that
\begin{equation*}
\mathcal{M}(\mcV_\eps \varrho_\eps) = \{\mcV \, \dd \mu\}.
\end{equation*}
\end{lemma}

\begin{proof}
    We use a density argument. Recall that
    \begin{align*}
        \mathcal{V}_\eps= \mathcal{V}^{\mathrm{W}}_\eps,\qquad \mathcal{V}(\varphi)= \langle \varphi^{\otimes 3}, \tilde {\mathcal{V}} \varphi^{\otimes 3} \rangle  ,\qquad \tilde{\mcV} := \mcS_3 \overline{\mcV},
    \end{align*}
    where $\mcS_3$ is the symmetrisation operator and
\begin{equation*}
\left(\overline {\mathcal{V}} \varphi^{\otimes 3}\right) (x,y,z) = v(x-y) v(x-z) \varphi(x) \varphi(y)\varphi(z).
\end{equation*}
Denote by 
      \begin{align*}
          H^{(s)}_{0,p}:= \sum_{j=1}^p h^{(s)}_j, \qquad   h^{(s)}_j:= 1\otimes \cdots  \otimes \underbrace{h^{(s)}}_{\text{jth position}} \otimes \cdots \otimes 1.
      \end{align*}  
      Also denote by $\Xi_R = \Xi(\frac{x}{R})$, where $\Xi$ is a smooth cut-off function which is one on $[0,1]$ and zero outside of $[-2,2]$. We define the regularised symbol $\tilde {\mathcal{V}}_R$ 
    \begin{align*}
        \tilde {\mathcal{V}}_R:= \Xi_R (H_{0,3}^{(s)}) \tilde{\mathcal{V}} = \Xi_R (H_{0,3}^{(s)})(H_{0,3}^{(s)}+1) (H_{0,3}^{(s)}+1)^{-1}   \tilde{\mathcal{V}} \in \mcL^{\infty}(\mfh^{\otimes_3}),
    \end{align*}
    since $(H_{0,3}^{(s)}+1)^{-1} $ is a compact operator and the others are bounded operators.   Using the Weyl commutation relations \eqref{Weyl_CR}, it follows from \cite[Proposition 2.9]{AN08} that
    \begin{equation*}
        [W_{\eps}(f),\mcV_\eps] = \left(\mcV\left(u + \frac{\ii \eps f}{\sqrt{2}}\right) - \mcV(u)\right)^{\mathrm{W}} = \eps B_\eps,
    \end{equation*}
    where $B_\eps$ can be bounded using powers of the number operator. Therefore it suffices to prove the result for $\mathrm{Tr}(W_\eps(f)\varrho_\eps \mcV_\eps)$.
    We write 
    \begin{multline}
   \left \vert  {\rm Tr } \left(  W_{\eps}(f) \varrho_{\eps}   \ \mathcal{V}^{\mathrm W}\right) - \int_{\mfh} \e^{\sqrt{2}\ii \pi  \mathrm{Re}\langle f,\varphi\rangle} \mathcal{V}(\varphi)  \, \dd \mu (\varphi)  \right \vert \leq  \left \vert   {\rm Tr } \left(  W_{\eps}(f)    \varrho_{\eps}  \left(  \mathcal{V}^{\mathrm W} -  \mathcal{V}_{R}^{\mathrm W} \right)\right) \right \vert 
\\ +\left \vert  {\rm Tr } \left(  W_{\eps}(f)    \varrho_{\eps}  \mathcal{V}_R^{\mathrm W}\right) - \int_{\mfh} \e^{\sqrt{2}\ii \pi  \mathrm{Re}\langle f,\varphi\rangle}  \tilde {\mathcal{V}}_R(\varphi)  \, \dd \mu (\varphi)  \right \vert
\\ + \left \vert \int_{\mfh} \e^{\sqrt{2}\ii \pi  \mathrm{Re} \langle f,\varphi\rangle} \left(  {\mathcal{V}}(\varphi)- {\mathcal{V}}_R(\varphi)\right)  \, \dd \mu (\varphi) \right \vert. \label{dominatedconvresult1} 
\end{multline} 
The second term on the right-hand side of \eqref{dominatedconvresult1} converges to $0$ because the symbol is compact. For the third term, one notes that we have pointwise convergence as $R \to \infty$ and both terms are bounded by $\|v\|_{L^\infty}^2\|\varphi\|_{\mfh}^4$, so convergence follows from the dominated convergence theorem. For the first term, we have 
\begin{multline}
\label{something}
   \left \vert   {\rm Tr } \left(  W_{\eps}(f)  \ \varrho_{\eps} \left( \mathcal{V}^{ \mathrm{W}} - \mathcal{V}_R^{ \mathrm{W}} \right)\right) \right \vert \leq \left \Vert (\mathcal{N}_\eps +1)^{\frac{3}{2}}  W_{\eps}(f) (\mathcal{N}_\eps +1)^{-\frac{3}{2}}   \right \Vert_{\mathcal{L}} \\
   \times \left \Vert (\mathcal{N}_\eps +1)^{\frac{3}{2}}  \varrho_\eps (H^{(s)}_{\eps,0} +1)^{\frac{3}{2}}   \right \Vert_{\mathcal{L}^1} 
    \left \Vert (H^{(s)}_{\eps,0}  +1)^{-\frac{3}{2}}  \left( \mathcal{V}^{\mathrm{W}} - \mathcal{V}_R^{ \mathrm{W}} \right) (\mathcal{N}_\eps +1)^{-\frac{3}{2}}   \right \Vert_{\mathcal{L}} .
\end{multline}
  with  {$s \in (0,1/2)$.}  The first term is bounded by \cite[Proposition 2.4]{NierAmmari}. The second term satisfies
  \begin{align*}
      \left \Vert (\mathcal{N}_\eps +1)^{\frac{3}{2}}  \varrho_\eps (H^{(s)}_{\eps,0} +1)^{\frac{3}{2}}   \right \Vert_{\mathcal{L}^1} \leq C (K+1)^{\frac{3}{2}},
  \end{align*}
  where $C $ is independent of $\eps$ by \eqref{asslemma}. Here we have also used the support properties of $\chi$. We now consider the third term on the right hand side of \eqref{something}. Setting $\tilde \Xi_R(\cdot):= 1- \Xi_R(\cdot)$, we can write  
  \begin{align*}
  \tilde {\mathcal{V}}- \tilde {\mathcal{V}}_R  = \tilde\Xi_R (H_{0,3}^{(s)}) \tilde {\mathcal{V}}.
  \end{align*}
Let $\Phi,\Psi \in \mathcal{D}$ for a dense domain $\mathcal{D}$. Then we have 
    \begin{align}
    \nonumber
       &  \langle \Phi, \left( \mathcal{V}-\mathcal{V}_R \right)^{\mathrm{W}}  \Psi \rangle 
         \\ 
         \nonumber
         &= \sum_{n\in \mathbb{N}} \langle \Phi^{(n)}, \left( \mathcal{V}-\mathcal{V}_R \right)^{\mathrm{W}}  \Psi^{(n)} \rangle 
         \\ 
         \nonumber
         &= \sum_{n \geq 3}   \eps^3  n(n-1)(n-2)\left \langle \Phi^{(n)}, \mathcal{S}_n\left(  \left(\tilde\Xi_R (H_{0,3}^{(s)}) \tilde{  \mathcal{V}} \right) \otimes 1^{(n-3)}\right)\Psi^{(n)} \right \rangle 
          \\ 
          \nonumber
          &= \sum_{n \geq 3}   \eps^3  n(n-1)(n-2)\left \langle \Phi^{(n)},  \left ( H^{(s)}_{0,3} \right)^{\frac{1}{2}} \left( H^{(s)}_{0,3}\right)^{-\frac{1}{2}}\left( \tilde\Xi_R (H_{0,3}^{(s)}) \tilde{  \mathcal{V}}  \otimes 1^{(n-3)}\right)\Psi^{(n)} \right \rangle 
          \\ 
          \label{horrid_fock_1}
          & \leq  \sum_{n \geq 3}   \eps^3  n(n-1)(n-2) \left\Vert     \left ( H^{(s)}_{0,3} \right)^{\frac{1}{2}}   \Phi^{(n)}\right \Vert_{\mfh^{(n)}}   \left \Vert   \left( H^{(s)}_{0,3}\right)^{-\frac{1}{2}} \tilde\Xi_R (H_{0,3}^{(s)})   \right \Vert_{\mathcal{L} (\mathfrak{h}^{(n)} )} \left \Vert \tilde{\mathcal{V}} \right \Vert_{\mathcal{L} (\mathfrak{h}^{(n)} )}  \left\Vert    \Psi^{(n)}\right \Vert_{\mfh^{(n)}}
    \end{align}
    where we have used Cauchy-Schwarz in the above computation.  Here $\mcS_n$ is the symmetrisation operator. We have
    \begin{align}
    \nonumber
     \eps^3 n(n-1)(n-2) & \left\Vert     \left ( H^{(s)}_{0,3} \right)^{\frac{1}{2}}  \Phi^{(n)}\right \Vert_{\mfh^{(n)}}^2  \\
     \nonumber
     & =     \eps^3 n(n-1)(n-2) \left \langle \Phi^{(n)} ,    H^{(s)}_{0,3}    \Phi^{(n)} \right \rangle 
    \\ 
    \nonumber 
    &=   \eps^3 n(n-1)(n-2) \left \langle \Phi^{(n)} ,   \left ( h^{(s)}_1+ h^{(s)}_2 +h^{(s)}_3 \right)   \Phi^{(n)} \right \rangle 
    \\ 
    \nonumber
    & {\leq}  3 \eps^3 n(n-1)(n-2) \left \langle \Phi^{(n)} ,  h^{(s)}_1 h^{(s)}_2 h^{(s)}_3   \Phi^{(n)} \right \rangle
    \\ 
    \nonumber
    &\lesssim \left \langle \Phi^{(n)} , \sum_{ 1\leq i_1,i_2, i_3\leq n} \eps ^3  h^{(s)}_{i_1}  h^{(s)}_{i_2}  h^{(s)}_{i_3}   \Phi^{(n)} \right \rangle 
    \\ 
    \label{horrid_fock}
    & \leq \left \Vert \left( H_{\eps,0}^{(s)}\right)^{3/2}   \Phi^{(n)} \right \Vert^2 \leq \left \Vert  \left( H_{\eps,0}^{(s)}+1\right)^{3/2}   \Phi^{(n)} \right \Vert^2
    \end{align}
    where we have used  $ h^{(s)}\geq 1$, that $ x_1+x_2+x_3 \leq 3 x_1 x_2 x_3$ for $x_i\geq 1$, and that $\Phi^{(n)}$ is symmetric. Similarly 
    \begin{align}
    \label{horrid_fock_2}
     \eps^3 n(n-1) (n-2) \left\Vert    \Psi^{(n)}\right \Vert_{\mfh^{(n)}}^2
  \leq   \left \Vert  \left(\mathcal{N}_\eps+1\right)^{3/2}   \Psi^{(n)} \right \Vert^2.
    \end{align}
    By the spectral theorem, we have 
    \begin{align}
    \label{spectral_theorem}
         \left \Vert   \left( H^{(s)}_{0,3}\right)^{-\frac{1}{2}} \tilde\Xi_R (H_{0,3}^{(s)})   \right \Vert_{\mathcal{L} (\mathfrak{h}^{(n)} )} \leq \frac{1}{\sqrt{R}}.
    \end{align}
   Combining \eqref{horrid_fock_1}--\eqref{spectral_theorem}, we get 
   \begin{align*}
    & \langle \Phi,  \left( \mathcal{V}-\mathcal{V}_R \right)^{\mathrm{W}}  \Psi \rangle_{\mcF}
          \\ & \leq  \sum_{n\geq 3} \eps^3  n(n-1)(n-2) \left\Vert     \left ( H^{(s)}_{0,3} \right)^{\frac{1}{2}}   \Phi^{(n)}\right \Vert_{\mfh^{(n)}}   \left \Vert   \left( H^{(s)}_{0,3}\right)^{-\frac{1}{2}} \tilde\Xi_R (H_{0,3}^{(s)})   \right \Vert_{\mathcal{L} (\mathfrak{h}^{(n)} )} \left \Vert \tilde{\mathcal{V}} \right \Vert_{\mathcal{L} (\mathfrak{h}^{(n)} )}  \left\Vert    \Psi^{(n)}\right \Vert_{\mfh^{(n)}}
          \\ & \lesssim \frac{K^{\frac{3}{2}} }{\sqrt{R}}  \left \Vert  \left(\mathcal{N}_\eps+1\right)^{3/2}   \Psi \right \Vert \left \Vert  \left( H_{\eps,0}^{(s)}+1\right)^{3/2}   \Phi \right \Vert \Vert v\Vert^2.
    \end{align*}
This implies that 
\begin{align*}
 \left \Vert   \left( H_{\eps,0}^{(s)}+1\right)^{-3/2}  \left( \mathcal{V}-\mathcal{V}_R \right)^{\mathrm{W}}  \left(\mathcal{N}_\eps+1\right)^{-3/2} \right \Vert_{\mathcal{L}} \lesssim  \frac{\Vert v\Vert^2 K^{\frac{3}{2}}  }{\sqrt{R}} \to 0
\end{align*}
as $R \to \infty$, which completes the proof.   
\end{proof}}

\begin{remark}
\label{remark_number_bounds_powers}
We note that we stated Lemma \ref{Wigner_measure_V_lemma} for states truncated with respect to the number operator. This can be relaxed to states that can absorb any power of the number operator.
\end{remark}

We also have a similar problem in that the symbols we need to consider for $\Theta_\eps(\xi)$ are not compact. To deal with this issue, we prove the following lemma. We first adopt the following notation
\begin{equation}
\label{Cp_B_p_definition}
\mathcal{C}_p := \mfBp \cup \{I_p\}, \quad \mathfrak{B}_p := \{\xi \in \mcL^2(\mfh^{(p)}) : \|\xi\|_{\mcL^2(\mfh^{(p)})} \leq 1\}.
\end{equation}

  \begin{lemma}
  \label{theta_convergence_lemma}
   Let $p \in \N_0$ and $s \in (0,1/2)$. Suppose that $\varrho_\eps$ is a family of trace-class quantum operators such that $\varrho_\eps = \varrho_\eps \chi(\mcN_\eps)$, $[\varrho_\eps,\mcN_\eps] = 0$, and
   \begin{align*}
       \| \varrho_\eps (H_{\eps,0}^{(s)} + 1)^{\frac{p}{2}}\|_{\mc{L}^1} < C(p),
   \end{align*}
uniformly in $\eps$. Moreover suppose that
\begin{equation*}
\mathcal{M}(\varrho_\eps) = \{\nu\}.
\end{equation*} 
Then for any fixed  $\xi \in \mcCp$,  one has 
   \begin{align*}
       \mathrm{Tr}(\Theta_\eps(\xi) \varrho_\eps) \underset{\eps \to 0}{ \longrightarrow} \int  \,  \Theta(\xi) \, \dd \nu .
   \end{align*}
\end{lemma}
\begin{proof}
The proof is analogous to the proof of Lemma \ref{Wigner_measure_V_lemma}, except we cut-off in $H_{0,p}^{(s)}$ and insert $(H_{\eps,0}^{(s)} + 1)^{\pm \frac{p}{2}}$. Here we consider the regularised symbol $\tilde{\xi}_R$ given by
\begin{align*}
        \tilde {\xi}_R:= \Xi_R (H_{0,p}^{(s)}) \tilde{\xi} = \Xi_R (H_{0,p}^{(s)})(H_{0,p}^{(s)}+1) (H_{0,p}^{(s)}+1)^{-1}   \tilde{\xi} \in \mcL^{\infty}(\mfh^{\otimes_p})
    \end{align*}
Note that here we do not need to commute $\Theta_\eps(\xi)$ and $W_\eps(f)$. We omit the details.
\end{proof}

\subsection{Miscellaneous Estimates}
In this section, we will collect some miscellaneous bounds which will be used throughout the paper. We have the following lemma which bounds the commutator of a Weyl operator with smooth $f$ and the free Hamiltonian.
\begin{lemma}
\label{free_number_bound}
Suppose that $f \in C^\infty(\T)$. Then one has the following bound
 \begin{equation*}
\|[W_\eps(f),H_{\eps,0}] (\mcN_\eps + 1)^{-\frac{1}{2}}\|_{\mc{L}} \leq \eps \sqrt{2} \|f\|_{H^2}+ \frac{\eps^2}{2} \|f\|_{H^1}^2.
\end{equation*}

\end{lemma}
\begin{proof}
    Note that the commutator $[W_\eps(f), H_{\eps,0}]$ can be written as 
    \begin{align*}
        [W_\eps(f), H_{\eps,0}]=  W_\eps (f) \Big( \eps \mathcal{R}_1(f) +\eps^2\mathcal{R}_2(f) \Big )^{\rm W}.
    \end{align*}
   Indeed, by using 
\begin{align*}
[W_\eps(f), H_{\eps,0}] & = W_\eps (f) \Big( H_{\eps,0}- W_\eps (f)^* H_{\eps,0} W_\eps (f) \Big)
\\ & =  W_\eps (f) \Big( h(\cdot)- h(\cdot + \frac{\ii \eps}{\sqrt{2}}  f) \Big )^{\rm W}.
\end{align*}
where we have used that 
\[  W_\eps (f)^* H_{\eps,0} W_\eps (f)= W_\eps (f)^* h^{\mathrm{W}} W_\eps (f) = \Big(h(\cdot + \frac{\ii \eps}{\sqrt{2}}  f) \Big )^{\rm W},\]
see \cite[Proposition 2.10]{AN08}. By  expanding, we get 
\begin{align}\label{taylorexpansion}
{h(\varphi)-h\left(\varphi+ \frac{\ii \eps}{\sqrt{2}}  f\right)}=\sum_{j=1}^2  \eps^{j}\mathcal{R}_j(f)[\varphi].
\end{align}
with 
\begin{align}
&\mathcal{R}_1(f)[\varphi]:=  \sqrt{2} \Im m \langle \varphi, hf \rangle \in \mathcal{P}_{0,1}+ \mathcal{P}_{1,0},
\\ &\mathcal{R}_2(f):= -\frac{1}{2} \langle f, hf \rangle  .
\end{align}
 Indeed,  the associated operator in  $\mathcal{R}_1(f)[\varphi] $ is
\begin{align}
 &  \widetilde{\mathcal{R}_1(f)}:\lambda \in \C \rightarrow \lambda h f \in {\rm Span}(hf)\subset \mfh,\qquad    \widetilde{\mathcal{R}_1(f)}  \in \mathcal{L}^\infty(\C, \mfh)
\end{align}
since it is finite-rank operator and  
\begin{align*}
   \left \Vert  \widetilde{\mathcal{R}_1(f)}(\lambda) \right \Vert_{\mfh}= \vert \lambda \vert \left \Vert hf  \right \Vert_{\mfh}=\vert \lambda \vert \left \Vert f  \right \Vert_{H^2}.
\end{align*}
Then we get
\begin{align*}
    \|  \Big(  \mathcal{R}_1(f)  \Big )^{\rm W}(\mcN_\eps + 1)^{-\frac{1}{2}}\|_{\mc{L}} \leq \sqrt{2}\| \widetilde{\mathcal{R}_1(f)}\|_{\mathcal{L}(\C ;\mfh)} =  \sqrt{2} \|f\|_{H^2}.
\end{align*}
Similarly, we get 
\begin{align*}
    \|  \Big(  \mathcal{R}_2(f)  \Big )^{\rm W}(\mcN_\eps + 1)^{-\frac{1}{2}}\|_{\mc{L}} \leq \frac{\|f\|^2_{H^1}}{2}.
\end{align*}
Therefore, we have 
\begin{equation*}
\|[W_\eps(f),H_{\eps,0}] (\mcN_\eps + 1)^{-\frac{1}{2}}\|_{\mc{L}} \leq \eps \sqrt{2} \|f\|_{H^2}+ \frac{\eps^2}{2} \|f\|_{H^1}^2. 
\end{equation*}
\end{proof}
We also have the following bounds, which follow from \cite{AN08}.
\begin{lemma}
\label{interaction_number_bound}
Suppose that $v \in L^\infty$. Then 
\begin{equation*}
\|\mc{V}_\eps (\mcN_\eps + 1)^{-3}\|_{\mcL} \leq C\|v\|_{L^\infty}^2 < \infty,
\end{equation*}
uniformly in $\eps$. Moreover, suppose that $\xi \in \mcL(\mfh^{(p)})$, then
\begin{equation*}
\|\Theta_\eps(\xi)(\mcN_\eps + 1)^{-p}\|_{\mcL} \leq C(p) < \infty,
\end{equation*}
uniformly in $\eps$.
\end{lemma}
Finally, we collect some bounds for the free Gibbs state. The first is a consequence of the quantum Wick theorem, see Lemma \ref{quantum_Wick_theorem_lemma}.
\begin{lemma}
\label{number_operator_powers_lemma}
Let $p \in \N_0$. Then one has
\begin{equation*}
\frac{1}{Z_{\eps,0}} \mathrm{Tr}(\mcN_\eps^p \e^{-H_{\eps,0}}) < C(p) < \infty,
\end{equation*}
uniformly in $\eps$.
\end{lemma}
We are also able to absorb  lifted  the fractional Laplacian for small powers.
\begin{lemma}
\label{free_Hamiltonian_1_lemma}
Suppose that $s \in (0,1/2)$ and let $m \in \R_{\geq 1}$. One has the bound
\begin{equation*}
\frac{1}{Z_{\frac{1}{m}\eps,0}} \mathrm{Tr}(H^{(s)}_{\eps,0} \e^{-\frac{1}{m}H_{\eps,0}}) < C(m,s) < \infty,
\end{equation*}
uniformly in $\eps$.
\end{lemma}
\begin{proof}
Denote 
\begin{equation*}
\rho_{\eps,0} := \frac{1}{Z_{\eps,0}}\e^{-H_{\eps,0}}, \qquad \rho_{\eps,0,m} := \frac{1}{Z_{\frac{1}{m}\eps,0}}\e^{-\frac{1}{m}H_{\eps,0}} .
\end{equation*}
Using \cite[Lemma 2.1]{LNR15}, we have
\begin{align*}
\rho_{\eps,0}^{(k)}= \left[ \frac{1}{\e^{\eps h} -1}\right]^{\otimes k}, \qquad \rho_{\eps,0,m}^{(k)}= \left[ \frac{1}{\e^{\frac{\eps h}{m}} -1}\right]^{\otimes k},
\end{align*}
where $\rho_{\eps,0}^{(k)}$ is the $k^{th}$ reduced density matrix. A direct computation leads to 
\begin{align}
\nonumber
\frac{1}{\mathrm{Tr}(\e^{-\frac{1}{m}H_{\eps,0}})} \mathrm{Tr}\left(H_{\eps,0}^{(s)}\e^{-\frac{1}{m}H_{\eps,0}}\right)
 & = \mathrm{Tr} \left(  \dd \Gamma_\eps ((-\Delta)^s + 1))
   \rho_{\eps,0,m}\right) 
   \\ \nonumber & = \mathrm{Tr} \left( \eps ((-\Delta)^s + 1)
   \rho_{\eps,0,m}^{(1)}\right)=  \mathrm{Tr} \left( \eps ((-\Delta)^s + 1)
   \frac{1}{\e^{\frac{\eps h}{m}} -1} \right)
   \\ \nonumber & \leq \mathrm{Tr} \left( \eps h^s
   \frac{1}{\e^{\frac{\eps h}{m}} -1} \right)+ \mathrm{Tr} \left( 
   \frac{\eps}{\e^{\frac{\eps h}{m}} -1} \right)
   \\ \label{bounded_Hamiltonian_equation} & \leq m\mathrm{Tr} \left(  {   h^{s-1}}
   \right)+ m\mathrm{Tr} \left( 
   h^{-1} \right)
 \leq mC ,
  \end{align}
  since $0<s<1/2$ and where we recall that $h=-\Delta +1.$
\end{proof}
Using the quantum Wick theorem, we also have the following lemma about powers of the lifted fractional Laplacian. We delay the proof until Appendix \ref{estimates_appendix}.
\begin{proposition}
\label{p-Hamiltonian_uniform_proposition}
Suppose that $p \in \N_0$ and let $s \in (0,\frac{1}{2})$. Then
\begin{equation}
\label{p_Hamiltonian_uniform_equation}
\frac{1}{Z_{\eps,0}}\mathrm{Tr}(\left(\dd \Gamma_{\eps}((-\Delta)^s)\right)^p\e^{-H_{\eps,0}}) < C(p,s) < \infty
\end{equation}
uniformly in $\eps > 0$.
\end{proposition}
We thus get the following corollary, which follows from Proposition \ref{p-Hamiltonian_uniform_proposition} by expanding the bracket and using Young's multiplication inequality.
\begin{corollary}
\label{p-Hamiltonian_uniform_corollary}
Suppose that $p \in \N_0$ and let $s \in (0,\frac{1}{2})$. Then
\begin{equation*}
\frac{1}{Z_{\eps,0}}\mathrm{Tr}\left(\left(H_{\eps,0}^{(s)} + 1\right)^p\e^{-H_{\eps,0}}\right) < C(p,s) < \infty
\end{equation*}
uniformly in $\eps > 0$.
\end{corollary}
\begin{remark}
We note that since $\e^{-H_{\eps,0}}$ and the number operator and lifted fractional Laplacians are positive, Lemmas \ref{number_operator_powers_lemma} and \ref{free_Hamiltonian_1_lemma}, as well as Proposition \ref{p-Hamiltonian_uniform_proposition} and Corollary \ref{p-Hamiltonian_uniform_corollary} also hold for the truncated free state.
\end{remark}
\section{Proof for bounded potentials}
\label{SEC:bounded}
In this section, we consider the problems for bounded interaction potentials.
\subsection{Duhamel expansion}
In this section, we recall the two power series expansions we will use throughout the rest of the paper. We identify an operator with its integral kernel, and we recall that
\begin{equation}
\label{curly_C_definition}
\mathfrak{B}_p := \{\xi \in \mc{L}^2(\mfh^{(p)}) : \|\xi\|_{\mcL^2(\mfh^{(p)})} \leq 1\}, \quad \mc{C}_p := \mf{B}_p \cup \{I_p\}.
\end{equation}
We note that we can write
\begin{equation}
\label{quantum_state_rewrite}
\rho_\eps\left(\Theta_\eps(\xi)\right) = \frac{\tilde{\rho}_{\eps,1}(\Theta_\eps(\xi))}{\tilde{\rho}_{\eps,1}(1)},
\end{equation}
where we define
\begin{equation*}
\tilde{\rho}_{\eps,z}(\mc{A}) := \frac{1}{Z_{\eps,0}} \mathrm{Tr}\left( \mc{A} \e^{-(H_{\eps,0} - z \mcV_{\eps})} \chi(\mcN_\eps)\right)
\end{equation*}
Define $A^\xi_\eps(z) := \tilde{\rho}_{\eps,z}(\Theta_\eps(\xi))$. Via Duhamel expansion, we have the following lemma, see for example \cite[Lemma 3.4]{RS22}.
\begin{lemma}
\label{quantum_Duhamel_expansion_lemma}
Suppose that $M \in \N$. Then we have $A^{\xi}_{\eps}(z) = \sum^{M-1}_{m=0} \alpha_{\eps,m}^\xi z^m + R^\xi_{\eps,M}(z)$, where
\begin{multline}
\label{quintic_quantum_explicit_terms}
\alpha_{\eps,m}^\xi := \frac{1}{Z_{\eps,0}} \mathrm{Tr}\biggl( \int_0^1 \dd t_1 \int_0^{t_1} \dd t_2 \ldots \int_0^{t_{m-1}} \dd t_m \, \Theta_\eps(\xi) \e^{-(1-t_1) H_{\eps,0}} \mcV_{\eps} \e^{-(t_1-t_2) H_{\eps,0}} \mcV_{\eps} \\
\times \e^{-(t_2-t_3)H_{\eps,0}} \ldots \e^{-(t_{m-1} -t_m) H_{\eps,0}} \mcV_{\eps} \e^{-t_m H_{\eps,0}}\chi(\mcN_{\eps})\biggr),
\end{multline}
and
\begin{multline}
\label{quintic_quantum_remainder_term}
R_{\eps,M}^\xi (z) := \frac{1}{Z_{\eps,0}} \mathrm{Tr}\biggl( \int_0^1 \dd t_1 \int_0^{t_1} \dd t_2 \ldots \int_0^{t_{M-1}} \dd t_M \, \Theta_\eps(\xi) \e^{-(1-t_1) H_{\eps,0}} \mcV_{\eps}
\e^{-(t_1-t_2) H_{\eps,0}} \mcV_{\eps} \\
\times \e^{-(t_2-t_3)H_{\eps,0}} \ldots \e^{-(t_{M-1} -t_M) H_{\eps,0}} \mcV_{\eps} \e^{-t_M (H_{\eps,0} - z\mcV_{\eps})} \chi(\mcN_{\eps})\biggr).
\end{multline}
\end{lemma}
Similarly, we can write
\begin{equation}
\label{classical_state_rewrite}
\rho(\Theta(\xi)) = \frac{\tilde{\rho}_1(\Theta(\xi))}{\tilde{\rho}_1(1)},
\end{equation}
where we define
\begin{equation*}
\tilde{\rho}_z(X) := \int \dd \mu_0 \, X \e^{z\mcV} \chi(\mcN).
\end{equation*}
Define $A^{\xi}(z) := \tilde{\rho}_z(\Theta(\xi))$. One can compute the following formal power series expansion.

\begin{lemma}
\label{classical_expansion_lemma}
Let $M \in \N$. Then $A^\xi(z) = \sum_{m=0}^{M-1} \alpha_m^\xi z^m + R_M^\xi(z)$, where
\begin{align}
\label{quintic_classical_explicit_terms}
\alpha^\xi_m &:= \frac{1}{m!} \int \dd \mu_0 \, \Theta(\xi) \mcV^m \chi(\mcN), \\
\label{quintic_classical_remainder_term}
R^\xi_M(z) &:= \frac{1}{M!} \int \dd \mu_0 \, \Theta(\xi) \e^{\tilde{z}\mcV} \mcV^M \chi(\mcN),  
\end{align}
for some $\tilde{z} \in [0,z]$.
\end{lemma}
We now recall some bounds from \cite{RS23}, whose proof follows verbatim for the rough cut-off, see \cite[Lemmas 3.8 and 3.10]{RS23}.
\begin{lemma}
\label{bounds_quantum_lemma}
Suppose that $v \in L^\infty(\T)$ and $\xi \in \mc{C}_p$. Then one has the following bounds.
\begin{align*}
|\alpha_{\eps,m}^\xi| &\leq \frac{K^p(K^3\|v\|_{L^\infty}^2)^m\|\xi\|}{m! \, 3^m}, \\
|R_{\eps,M}^\xi(z)| &\leq \e^{\frac{1}{3} K^3 |\mathrm{Re}(z)| \|v\|_{L^\infty}^2 }\frac{K^p(K^3\|v\|_{L^\infty}^2)^M\|\xi\|}{M! \, 3^M} |z|^M.
\end{align*}
\end{lemma}
We also have the corresponding bounds for the classical expansion, see \cite[Lemmas 3.12 and 3.13]{RS23}.
\begin{lemma}
\label{bounds_classical_lemma}
Suppose that $v \in L^\infty(\T)$ and $\xi \in \mc{C}_p$. Then one has the following bounds.
\begin{align*}
|\alpha_{m}^\xi| &\leq \frac{K^p(K^3\|v\|_{L^\infty}^2)^m\|\xi\|}{m! \, 3^m}, \\
|R_{M}^\xi(z)| &\leq \e^{\frac{1}{3} K^3 |\mathrm{Re}(z)| \|v\|_{L^\infty}^2 }\frac{K^p(K^3\|v\|_{L^\infty}^2)^M\|\xi\|}{M! \, 3^M} |z|^M.
\end{align*}
\end{lemma}
{\begin{remark}
Let us remark that by arguing as in the proof of \cite[Lemma 3.8]{RS23}, one can prove that
\begin{multline}
\label{non_integrated_bound}
\frac{1}{Z_{\eps,0}} \biggl|\mathrm{Tr}\biggl( \Theta_\eps(\xi) \e^{-(1-t_1) H_{\eps,0}} \mcV_{\eps} \e^{-(t_1-t_2) H_{\eps,0}} \mcV_{\eps} \e^{-(t_2-t_3)H_{\eps,0}} \ldots \\
\times \e^{-(t_{m-1} -t_m) H_{\eps,0}} \mcV_{\eps} \e^{-t_m H_{\eps,0}}\chi(\mcN_{\eps})\biggr)\biggr| \leq \frac{K^p(K^3\|v\|_{L^\infty}^2)^m\|\xi\|}{ 3^m}.
\end{multline}
\end{remark}}
We have the following corollary, which follows from Lemmas \ref{bounds_quantum_lemma} and \ref{bounds_classical_lemma}.
\begin{corollary}
\label{analytic_corollary}
The power series in Lemmas \ref{quantum_Duhamel_expansion_lemma} and \ref{classical_expansion_lemma} are analytic on $\C$.
\end{corollary}
\begin{remark}
The main difficulty of the proof of Theorem \ref{bounded_theorem} is proving that
\begin{equation}
\label{convergence_explicit_terms}
\alpha_{\eps,m}^\xi \to \alpha_m^\xi
\end{equation}
uniformly in $\xi \in \mc{C}_p$. The majority of this section is dedicated to the proof of \eqref{convergence_explicit_terms}.
\end{remark}

\subsection{Wigner measures of terms in Duhamel expansion}
In this section, we show that the Wigner measures of the explicit terms of the quantum expansion are the terms in the classical expansion. By combining this with the results in Section \ref{Section:non-compactness}, the convergence of the explicit terms will follow. We write
\begin{equation*}
\mathcal{U} := \{(t_1,\ldots,t_m) : 0 < t_m < \ldots < t_1 < 1\}.
\end{equation*}
We denote elements of $\mc{U}$ by $\mbt$. We also adopt the notation
\begin{equation*}
\rho_{\eps,m}(\mbt) :=  \frac{1}{Z_{\eps,0}} \e^{-(1-t_1) H_{\eps,0}} \mcV_{\eps} \e^{-(t_1-t_2) H_{\eps,0}} \mcV_{\eps} \e^{-(t_2-t_3)H_{\eps,0}} \ldots \e^{-(t_{m-1} -t_m) H_{\eps,0}} \mcV_{\eps} \e^{-t_m H_{\eps,0}}\chi(\mcN_{\eps}),
\end{equation*}
with the convention that
\begin{equation*}
\rho_{\eps,0}(\mbt) = \rho_{\eps,0} = \frac{1}{Z_{\eps,0}} \e^{-H_{\eps,0}}\chi(\mcN_\eps).
\end{equation*}
We begin by proving some estimates.
\begin{lemma}
\label{explicit_terms_number_lemma}
Let $m,p \in \N_0$ and $\mbt \in \mc{U}$. Then one has that
\begin{align}
\label{explicit_terms_number_equation}
\|(\mcN_\eps + 1)^p \rho_{\eps,m}(\mbt)\|_{\mc{L}^1} < C(m,p) < \infty,
\end{align}
uniformly in $\eps$. 
\end{lemma}
\begin{proof}
This follows by arguing as in the proof of \cite[Lemma 3.8]{RS23}, and using the truncation in the number operator.
\end{proof}
\begin{remark}
Lemma \ref{explicit_terms_number_lemma} implies that for each $m \in \N_0$ and $\mbt \in \mc{U}$, there exists a Wigner measure for $\rho_{\eps,m}(\mbt)$. Indeed, we will show that this Wigner measure is unique, uniform for $\mbt \in \mc{U}$, and given by the corresponding term in the classical expansion.
\end{remark}
\begin{lemma}
\label{explicit_terms_Hamiltonian_lemma}
Let $m,p \in \N_0$ and $\mbt \in \mc{U}$. Then one has that
\begin{align}
\label{explicit_terms_Hamiltonian_equation_1}
\|(H_{\eps,0}^{(s)} + 1)^p \rho_{\eps,m}(\mbt)\|_{\mc{L}^1} &< C(m,p,\mbt) < \infty, \\
\label{explicit_terms_Hamiltonian_equation_3}
\|\rho_{\eps,m}(\mbt) (H_{\eps,0}^{(s)} + 1)^p\|_{\mc{L}^1} &< C(m,p,\mbt) < \infty, \\
\label{explicit_terms_Hamiltonian_equation_2}
\|(H_{\eps,0}^{(s)} + 1)^\frac{p}{2} \rho_{\eps,m}(\mbt) (H_{\eps,0}^{(s)} + 1)^\frac{p}{2}\|_{\mc{L}^1} &< C(m,p,\mbt) < \infty.
\end{align}
uniformly in $\eps$.
\end{lemma}
\begin{proof}
For \eqref{explicit_terms_Hamiltonian_equation_1}, we argue similarly to \cite[Proof of Lemma 3.8]{RS23}.  Note that
\begin{equation}
\label{trace_identity}
\left\| \e^{-tH_{\eps,0}} \right\|_{\mc{L}^{\frac{1}{t}}} = \mathrm{Tr}(\e^{-H_{\eps,0}})^{t} = Z_{\eps,0}^t.
\end{equation}
Therefore one writes
\begin{multline*}
\|(H_{\eps,0}^{(s)} + 1)^p \rho_{\eps,m}(\mbt)\|_{\mc{L}^1}  \leq \frac{1}{Z_{\eps,0}}\|(H_{\eps,0}^{(s)} + 1)^p \e^{-(1-t_1)H_{\eps,0}}\|_{\mcL^{\frac{1}{1-t_1}}}  \prod_{j=1}^m \|\e^{-(t_j-t_{j+1})H_{\eps,0}}\|_{\mc{L}^{\frac{1}{t_j-t_{j+1}}}}
\\  \times \|\mcV_\eps \chi(\mcN_\eps)\|_{\mc{L}}^m\lesssim_m \frac{Z_{\eps,0}^{t_1}}{Z_{\eps,0}} \|(H_{\eps,0}^{(s)} + 1)^p \e^{-(1-t_1)H_{\eps,0}}\|_{\mcL^{\frac{1}{1-t_1}}} ,
\end{multline*}
where we adopt the convention $t_{m+1} = 0$. We have also used that every operator in the definition of $\rho_{\eps,m}(\mbt)$ commutes with the number operator. By Corollary \ref{p-Hamiltonian_uniform_corollary}, the one has that
\begin{equation}
\label{intermediate_hamiltonian_estimate}
\|(H_{\eps,0}^{(s)} + 1)^p \e^{-(1-t_1)H_{\eps,0}}\|_{\mcL^{\frac{1}{1-t_1}}} = \left[\mathrm{Tr}\left( (H_{\eps,0}^{(s)} + 1)^{\frac{p}{1-t_1}} \e^{-H_{\eps,0}} \right)\right]^{{1-t_1}} \lesssim C(\mbt,p) Z_{\eps,0}^{1-t_1}.
\end{equation}
By \eqref{trace_identity} and \eqref{intermediate_hamiltonian_estimate}, it follows that 
\begin{equation*}
\|(H_{\eps,0}^{(s)} + 1)^p \rho_{\eps,m}(\mbt)\|_{\mc{L}^1} \lesssim_{m,p,\mbt} \frac{Z_{\eps,0} ^{1-t_1}}{Z_{\eps,0} ^{1-t_1}}.
\end{equation*}
The other two estimates follow similarly, so we omit the details.
\end{proof}
We will use the following commutator bound.
\begin{lemma}
\label{Weyl_commutator_bound_lemma}
Suppose that $f \in L^2(\T)$, $\mbt \in \mathcal{U}$, and let $m \in \N$. Then
\begin{equation}
\label{ugly_commutator}
\lim_{\eps \to 0} \frac{1}{Z_{\eps,0}}\mathrm{Tr}\left(\left[W_{\eps}(f),\e^{-(1-t_1)H_{\eps,0}}\right] \mcV_{\eps} \e^{-(t_1-t_2)H_{\eps,0}} \ldots \mcV_{\eps} \e^{-t_m H_{\eps,0}} \chi(\mcN_{\eps})\right) = 0.
\end{equation}
Moreover convergence is uniform in $\mbt$.
\end{lemma}
\begin{proof}
First let $f \in C^\infty(\R)$. We recall the formula 
\begin{equation*}
[A,\e^{-B}] = \int_0^1 \dd s\, e^{-s B} [A,B] \e^{-(1-s)B},
\end{equation*}
with $B = (1-t_1)H_{\eps,0}$. We thus rewrite \eqref{ugly_commutator} as
\begin{equation}
\label{ugly_commutator_rewrite_1}
\frac{1}{Z_{\eps,0}}\mathrm{Tr} \bigg( \int_0^1 \dd s  \, e^{-s B} [W_\eps(f),B] \e^{-(1-s)B} \mcV_{\eps} \e^{-(t_1-t_2)H_{\eps,0}} \ldots \mcV_{\eps} \e^{-t_m H_{\eps,0}} \chi(\mcN_{\eps}) \bigg).
\end{equation}
One can bound \eqref{ugly_commutator_rewrite_1} by
\begin{equation*}
\frac{1}{Z_{\eps,0}} \int_0^1 \left\|e^{-s B} [W_\eps(f),B] \e^{-(1-s)B} \chi(\mcN_\eps) \right\|_{\mc{L}^{\frac{1}{1-t_1}}} \|\mcV_\eps \chi(\mcN_\eps)\|_{\mc{L}}^m \prod_{j=1}^m \|\e^{-(t_j-t_{j+1})H_{\eps,0}}\|_{\mc{L}^{\frac{1}{t_j-t_{j+1}}}},
\end{equation*}
where we adopt the convention $t_{m+1} = 0$. Here we have used that the number operator commutes with every non-Weyl operator in \eqref{ugly_commutator_rewrite_1} and that $\chi^p = \chi$ for any $p \in \N$. Then using the support properties of $\chi$ and Lemma \ref{free_number_bound}, one has
\begin{multline*}
\left\|e^{-s B} [W_\eps(f),B] \e^{-(1-s)B} \chi(\mcN_\eps) \right\|_{\mc{L}^{\frac{1}{1-t_1}}} \\
\leq (1-t_1)\|\e^{-s(1-t_1)H_{\eps,0}}\|_{\mc{L}^{\frac{1}{s(1-t_1)}}} \|\e^{-(1-s)(1-t_1)H_{\eps,0})}\|_{\mc{L}^{\frac{1}{(1-s)(1-t_1)}}} \eps  (\|f\|_{H^2} + \|f\|_{H^1}^2) (K+1)^\frac{1}{2}.
\end{multline*}
Recalling \eqref{trace_identity} and using $0 \leq (1-t_1) \leq 1$, one can bound the absolute value of \eqref{ugly_commutator_rewrite_1} by
\begin{equation*}
\eps \frac{Z_{\eps,0}}{Z_{\eps,0}}\frac{C(\|f\|_{H^2})(K+1)^{3m + \frac{1}{2}}\|v\|_{L^\infty}^{2m}}{3^m} \to 0
\end{equation*}
as $\eps \to 0$. In particular, \eqref{ugly_commutator} holds for $f$ smooth. For $f \in \mfh$, we use a density argument. Take $f_n \to f$ with $f_n$ smooth. By \eqref{Weyl_CR}, one has
\begin{equation*}
W_\eps(f_n) - W_\eps(f) = W_\eps(f_n)(\e^{\frac{ \ii \eps}{2} \mathrm{Im} \langle f_n,f \rangle} W_\eps(f-f_n) - I).
\end{equation*}
Therefore, one has 
\begin{equation}
\label{ugly_commutator_2}
[W_\eps(f_n) - W_\eps(f),\e^{-(1-t_1)H_{\eps,0}}] = [\e^{\frac{ \ii \eps}{2} \mathrm{Im} \langle f_n,f \rangle} W_\eps(f_n)W_\eps(f_n-f), \e^{-(1-t_1)H_{\eps,0}}] - [W_\eps(f_n),\e^{-(1-t_1)H_{\eps,0}}].
\end{equation}
For fixed $n$, the second term in \eqref{ugly_commutator_2} will give zero in \eqref{ugly_commutator} since $f_n$ is smooth. One splits the first commutator in \eqref{ugly_commutator_2} into its two parts. For the first term, one needs to bound
\begin{equation}
\label{commutator_intermediate_expression_1}
\frac{ \e^{\frac{ \ii \eps}{2} \mathrm{Im} \langle f_n,f \rangle}}{Z_{\eps,0}}\mathrm{Tr} \bigg( W_\eps(f_n) W_\eps(f_n-f) \e^{-(1-t_1)H_{\eps,0}}  \mcV_{\eps} \e^{-(t_1-t_2)H_{\eps,0}} \ldots \mcV_{\eps} \e^{-t_m H_{\eps,0}} \chi(\mcN_{\eps}) \bigg).
\end{equation}
To estimate this, we argue similarly to the proof of \cite[Lemma 3.8]{RS23}. In particular, we can bound \eqref{commutator_intermediate_expression_1} by
\begin{equation*}
\frac{1}{Z_{\eps,0}} \|W_\eps(f_n){W_\eps(f_n-f)}\chi(\mcN_\eps)\|_{\mc{L}} \|\mcV_\eps \chi(\mcN_\eps)\|_{\mc{L}}^m \prod_{i=1}^m \|\e^{-(t_j-t_{j+1})H_{\eps,0}}\|_{\mc{L}^{\frac{1}{t_j-t_{j+1}}}}.
\end{equation*}
Here we adopt the notation $t_0 = 1$ and $t_{m+1} = 0$. Note that in making this argument, we have used that $\chi(\mcN_\eps)$ commutes with $H_{\eps,0}$ and $\mcV_\eps$. We then use the support properties of $\chi$ and \cite[Lemma 3.1]{AN08} to bound this above by
\begin{equation*}
\|f_n-f\|_{\mfh} \frac{Z_{\eps,0}}{Z_{\eps,0}} \frac{(K+1)^{3m + \frac{1}{2}} \|v\|_{L^\infty}^{2m}}{3^m}.
\end{equation*}
In particular, this term can be made arbitrarily small. We argue similarly for the term corresponding to $\e^{\frac{ \ii \eps}{2} \mathrm{Im} \langle f_n,f \rangle}\e^{-(1-t_1)_{H_{\eps,0}}}W_\eps(f_n)W_\eps(f_n-f)$. In this case, one uses cyclicity of the trace to avoid commuting the number operator and the Weyl operators. We omit the details.
\end{proof}
\begin{remark}
We have made heavy use of the truncation to prove \eqref{ugly_commutator}. We expect that \eqref{ugly_commutator} should also hold in the untruncated (defocusing) setting as well, however one likely needs to appeal to the diagrammatic arguments introduced in \cite{FKSS17}. Since the result is not necessary for the current paper, we do not comment further on the defocusing problem.
\end{remark}

\begin{proposition}
\label{convergence_proposition}
For any $m \in \N_0$ and any $\mbt \in \mathcal{U}$, we have that \begin{equation}
\label{convergence_intermediate}
\mc{M}(\rho_{\eps,m}(\mbt)) = \left\{ \mcV^m \chi(\mcN) \,  \dd \mu_0 \right\}.
\end{equation}
\end{proposition}
\begin{proof}
We argue \eqref{convergence_intermediate} by induction. Let us note that the $m=0$ case of the induction is precisely Lemma \ref{convergence_free_cutoff_measure_lemma}. We can write
\begin{multline*}
\mathrm{Tr}(W_{\eps}(f) \rho_{\eps,m}(\mbt)) = \frac{1}{Z_{\eps,0}}\mathrm{Tr}\left(\left[W_{\eps}(f),\e^{-(1-t_1)H_{\eps,0}}\right] \mcV_{\eps} \e^{-(t_1-t_2)H_{\eps,0}} \ldots \mcV_{\eps} \e^{-t_m H_{\eps,0}} \chi(\mcN_{\eps})\right) \\
+ \frac{1}{Z_{\eps,0}} \mathrm{Tr}(\e^{-(1-t_1)H_{\eps,0}} W_{\eps}(f) \mcV_{\eps} \e^{-(t_1-t_2)H_{\eps,0}} \mcV_{\eps} \ldots \e^{-t_m H_{\eps,0}}\chi(\mcN_{\eps})).
\end{multline*}
By Lemma \ref{Weyl_commutator_bound_lemma}, the first term will vanish in the limit $\eps \to 0$. We use cyclicity of the trace and the fact that the free Hamiltonian is particle number conserving to write
\begin{equation*}
\lim_{\eps \to 0} \mathrm{Tr}(W_{\eps}(f) \rho_{\eps,m}(\mbt)) = \lim_{\eps \to 0} \mathrm{Tr} (W_{\eps}(f)  \mcV_{\eps} \rho_{\eps,m-1}(\mbs)).
\end{equation*}
Here we have written $s_{i} := 1 - t_1 + t_{i+1}$. Using the inductive hypothesis, Lemma \ref{explicit_terms_Hamiltonian_lemma}, and Lemma \ref{Wigner_measure_V_lemma} 
\begin{equation*}
\lim_{\eps \to 0} \mathrm{Tr}(W_\eps(f)\rho_{\eps,m}(\mbt)) = \int \dd \mu_0  \, \e^{\sqrt{2} \ii \mathrm{Re}\langle f,u \rangle} \mc{V}^m \chi(\mcN),
\end{equation*}
so \eqref{convergence_intermediate} follows.
\end{proof}
\subsection{Convergence of the explicit terms}
In this section, we prove the convergence of the quantum explicit terms. We first prove the following pointwise convergence result.
\begin{lemma}[Pointwise convergence]
Suppose that $m,p \in \N_0$. Then for all $\xi \in \mcCp$ and any $\mbt \in \mc{U}$, one has
\begin{equation*}
\mathrm{Tr}(\Theta_\eps(\xi)\rho_{\eps,m}(\mbt)) \to \int \dd \mu_0 \, \mc{V}^m \Theta(\xi) {\chi(\mcN)}, 
\end{equation*}
as $\eps \to 0$. In particular, one has $\alpha_{\eps,m}^\xi \to \alpha_m^\xi$ as $\eps \to 0$.
\end{lemma}
\begin{proof}
This is a consequence of Lemma \ref{theta_convergence_lemma} and Lemma \ref{explicit_terms_Hamiltonian_lemma}. Convergence of the explicit terms then follows by the dominated convergence theorem.
\end{proof}
We have now established the pointwise convergence of the explicit terms for $\xi \in \mcCp$. To prove Theorem \ref{bounded_theorem}, we need this convergence to be uniform in $\mcCp$. To extend to uniform convergence, we first prove the following equicontinuity result.
\begin{lemma}
\label{equicontinuity_lemma}
Suppose that $\xi_1,\xi_2 \in \mfBp$. Then
\begin{equation*}
\left| \mathrm{Tr}\left(\Theta_\eps(\xi_1) \rho_{\eps,m}(\mbt)\right) -  \mathrm{Tr}\left(\Theta_\eps(\xi_2) \rho_{\eps,m}(\mbt)\right)\right| \lesssim_{m,p} \|\xi_1-\xi_2\|_{L^2}.
\end{equation*}
\end{lemma}
\begin{proof}
We note that by linearity of the trace and the definition of $\Theta_\eps(\xi)$, one has
\begin{equation*}
\left|\mathrm{Tr}\left(\Theta_\eps(\xi_1) \rho_{\eps,m}(\mbt)\right) -  \mathrm{Tr}\left(\Theta_\eps(\xi_2) \rho_{\eps,m}(\mbt)\right)\right| = \left| \mathrm{Tr}\left(\Theta_\eps(\xi_1 - \xi_2) \rho_{\eps,m}(\mbt)\right) \right| \lesssim K^{3m+p} \|v\|_{L^\infty}^{2m} \|\xi_1-\xi_2\|_{L^2},
\end{equation*}
where the final inequality follows from inserting powers of the number operator and using Lemma \ref{interaction_number_bound} and Lemma \ref{explicit_terms_number_lemma}.
\end{proof}
Similarly, one has the corresponding result for the classical object.
\begin{lemma}
\label{continuity_classucal_lemma}
Suppose that $\xi_1,\xi_2 \in \mfBp$. Then
\begin{equation*}
\left| \int \dd \mu_0 \, \left(\Theta(\xi_1) -\Theta(\xi_2)\right)\mcV^m \chi(\mcN) \right| \lesssim_{m,p} \|\xi_1-\xi_2\|_{L^2}.
\end{equation*}
\end{lemma}
\begin{proof}
By H\"older's inequality, one has
\begin{equation*}
|\mcV(u)| \leq \|v\|_{L^\infty}^2 \|u\|_{L^2}^6.
\end{equation*}
Using Lemma \ref{random_variable_bound_lemma} and the fact any power of the $L^2$ norm is finite because of the truncation, it follows that 
\begin{equation*}
\left| \int \dd \mu_0 \, \left(\Theta(\xi_1) -\Theta(\xi_2)\right)\mcV^m \chi(\mcN) \right| \lesssim K^{3m+p} \|v\|_{L^\infty}^{2m}\|\xi_1-\xi_2\|_{L^2}.
\end{equation*}
\end{proof}
We are now able to prove the uniform convergence of the quantum explicit terms to the classical explicit terms.
\begin{lemma}[Uniform convergence]
Recall the definition of $\mcCp$ from \eqref{curly_C_definition}. Then for any $p \in \N_0$ and $m \in \N$, one has $\alpha_{\eps,m}^\xi \to \alpha^\xi_m$ uniformly in $\xi \in \mcCp$ as $\eps \to 0$.
\end{lemma}
\begin{proof}
We note that it suffices to prove that convergence is uniform for $\xi \in \mfBp$. Let us note that the convergence of $\alpha^{\xi}_{ \eps,m} \to \alpha^{\xi}_m$ on the set
\begin{equation*}
\{\xi \in \mfBp : \mathrm{supp}(\hat{\xi}) \subset {[-N,N]^{2p}}\}
\end{equation*}
is uniform. Here we have used Lemma \ref{equicontinuity_lemma} and the fact that bounded finite dimensional subsets are compact. We now write
\begin{multline}
\label{ugly_uniform_convergence}
\left| \mathrm{Tr}(\Theta_\eps(\xi) \rho_{\eps,m}(\mbt)) - \int \dd \mu_0 \,  \Theta(\xi)\mcV^m\chi(\mcN)\right| \leq \left| \mathrm{Tr}(\Theta_\eps(\xi_N) \rho_{\eps,m}(\mbt)) - \int \dd \mu_0 \,  \Theta(\xi_N)\mcV^m\chi(\mcN)\right| \\ 
+ \left| \mathrm{Tr}(\Theta_\eps(\xi - \xi_N) \rho_{\eps,m}(\mbt)) \right| +  \left| \int \dd \mu_0 \,  \Theta(\xi-\xi_N)\mcV^m\chi(\mcN)\right|,
\end{multline}
where $\xi_N := \sum_{k=-N}^N \hat{\xi}(k) \e^{2 \pi \ii k x}$. From the earlier discussion, we have that the first term converges uniformly in $\xi \in \mfBp$. We fix $s \in (0,\frac{1}{2})$. For the third term on the right hand side of \eqref{ugly_uniform_convergence}, one notes that
\begin{multline*}
\left| \int \dd \mu_0 \,  \Theta(\xi-\xi_N)\mcV^m\chi(\mcN)\right| \\
\leq \int \dd \mu_0 \bigg| \int \prod_{j=1}^p \dd x_j \, \dd y_j \, (\xi-\xi_N)(\mbx;\mby) \prod_{j=1}^p ((-\Delta_{x_j})^s + 1)^{-\frac{1}{2}}\prod_{j=1}^p  ((-\Delta_{y_j})^s + 1)^{-\frac{1}{2}} \\
\prod_{j=1}^p  ((-\Delta_{x_j})^s + 1)^{\frac{1}{2}} \prod_{j=1}^p ((-\Delta_{y_j})^s + 1)^{\frac{1}{2}} \prod_{j=1}^p \overline{\vph(x_j)}\vph(y_j)  \mcV^m \chi(\mcN) \bigg| \\
\lesssim \|v\|_{L^\infty}^{2m} K^{3m} \mathbb{E}_{\mu_0}[\|\vph\|^{2p}_{H^s(\T)}] \|\xi - \xi_N\|_{H^{-s}(\T^{2p})}.
\end{multline*}
Here we use that elements of the free field are in $H^{s}(\T)$ almost surely. This converges uniformly to $0$ for $\xi \in \mfBp$. Finally, for the second term on the right hand side of \eqref{ugly_uniform_convergence}, we write 
\begin{equation*}
\Theta_\eps(\xi - \xi_N) \rho_{\eps,m}(\mbt) = \left(H_{\eps,0}^{(s)} + 1\right)^{\frac{p}{2}} \left(H_{\eps,0}^{(s)} + 1\right)^{-\frac{p}{2}} \Theta_{\eps}(\xi-\xi_N) \left(H_{\eps,0}^{(s)} + 1\right)^{-\frac{p}{2}} \left(H_{\eps,0}^{(s)} + 1\right)^{\frac{p}{2}} \rho_{\eps,m}(\mbt).
\end{equation*}
Using cyclicity of the trace, we have that the second term on the right hand side of \eqref{ugly_uniform_convergence} is
\begin{equation*}
\lesssim_{m,p,\mbt} \left\| \left(H_{\eps,0}^{(s)} + 1\right)^{-\frac{p}{2}} \Theta_{\eps}(\xi-\xi_N) \left(H_{\eps,0}^{(s)} + 1\right)^{-\frac{p}{2}}\right \|_{\mc{L}} 
\left\|(H_{\eps,0}^{(s)} + 1)^{\frac{p}{2}}\rho_{\eps,m}(\mbt)(H_{\eps,0}^{(s)} + 1)^{\frac{p}{2}}\right\|_{\mc{L}^1}.
\end{equation*}
Recalling Corollary \ref{quantum_tail_bound} and using Lemma \ref{explicit_terms_Hamiltonian_lemma}, one has that
\begin{equation}
\label{uniform_convergence_equation}
\left| \mathrm{Tr}(\Theta_\eps(\xi - \xi_N) \rho_{\eps,m}(\mbt)) \right| \lesssim_{m,p,\mbt}  \|\xi - \xi_{N}\|_{H^{-s}(\T^{2p})} \to 0
\end{equation}
as $N \to \infty$ uniformly on $\mfBp$. It follows that one has that $\alpha_{\eps,m}^\xi \to \alpha_m^\xi$ uniformly in $\mfBp$ as $\eps \to 0$.  Here we use the dominated convergence theorem. {Here one interpolates using \eqref{non_integrated_bound} to ensure integrability in the simplex.}
\end{proof}

\subsection{Proof of Theorem \ref{bounded_theorem}}
\label{Section:proof_of _bounded_theorem}
We first show the convergence of the quantum state to the classical state. Recall the definition of $\mc{C}_p$ in \eqref{curly_C_definition}. We have the following convergence result.
\begin{lemma}
$A_\eps^\xi(z) \to A^\xi(z)$ as $\eps \to 0$ uniformly in $\xi \in \mc{C}_p$.
\end{lemma}
\begin{proof}
We have that $A_\eps^\xi$ and $A^\xi$ are both analytic by Corollary \ref{analytic_corollary}. In particular, for any $z \in \C$, one has
\begin{equation*}
\sup_{\xi \in \mc{C}_p}\left|A_\eps^\xi(z) - A^\xi(z)\right| \leq \sum_m \sup_{\xi \in \mc{C}_p}\left|\alpha^\xi_{\eps,m} - \alpha^\xi_{m} \right| |z|^m \to 0
\end{equation*}
as $\eps \to 0$. Here we have applied Lemmas \ref{bounds_quantum_lemma} and \ref{bounds_classical_lemma} along with the dominated convergence theorem.
\end{proof}
We obtain the following corollary by recalling \eqref{quantum_state_rewrite} and \eqref{classical_state_rewrite}, and by setting $z=1$.
\begin{corollary}
\label{state_convergence_corollary}
For any $p \in \N_0$, one has $\rho_\eps(\Theta_\eps(\xi))$ converges to $\rho(\Theta(\xi))$ as $\eps \to 0$ uniformly in $\xi \in \mc{C}_p$.
\end{corollary}
\begin{proof}[Proof of Theorem \ref{bounded_theorem}]
Theorem \ref{bounded_theorem} now follows exactly as in the proof of \cite[Theorem 1.4]{RS23}, since the smoothness of the cut-off is not used. The proof is based on the observation that for $p \in \N_0$, one has
\begin{equation*}
\mathrm{Tr}(\gamma_\eps^p\xi) = \rho_\eps(\Theta_\eps(\xi)), \quad \mathrm{Tr}(\gamma^p\xi) = \rho(\Theta(\xi)),
\end{equation*}
Corollary \ref{state_convergence_corollary}, and a duality argument. The convergence of the quantum relative partition function follows similarly. For full details, we direct the reader to  \cite[Section 3.7 and the proof of Theorem 1.4]{RS23}
\end{proof}

\section{Unbounded interaction potentials}
\label{Section:unbounded_potentials}
In this section, we consider the sequence of interaction potentials defined in \eqref{approximation_identity}. In particular, we give a proof of Theorem \ref{delta_theorem}. We start with the following lemma.
\begin{lemma}
\label{state_convergence_delta}
Suppose that $v^N$ is defined as in \eqref{approximation_identity} and $v = \delta$. Suppose further that the cut-off function satisfies Assumption \ref{cut-off_assumption} with $K$ sufficiently small that the Gibbs measure for the focusing local quintic NLS is well defined. Then there is a sequence $N_\eps$ tending to $\infty$ as $\eps \to 0$ such that for any $p \in \N_0$, one has
\begin{equation}
\label{delta_state_convergence}
\lim_{\eps \to 0} \rho_\eps^{N_\eps}(\Theta_\eps(\xi)) = \rho(\Theta(\xi))
\end{equation}
uniformly in $\mc{C}_p$, where we recall \eqref{curly_C_definition}.
\end{lemma}
\begin{proof}
To prove \eqref{delta_state_convergence}, we note that by a diagonal argument, it suffices to prove that for fixed $N$, one has
\begin{equation}
\label{delta_1}
\lim_{\eps \to 0}\rho_\eps^N(\Theta_\eps(\xi)) = \rho^N(\Th(\xi))
\end{equation}
uniformly in $\xi \in \mc{C}_p$, and
\begin{equation}
\label{delta_2}
\lim_{N \to \infty} \rho^N(\Th(\xi)) = \rho(\Theta(\xi)).
\end{equation}
We note that \eqref{delta_1} is a consequence of $v^N$ being in $L^\infty$, so we can apply Theorem \ref{bounded_theorem}. By arguing as in \cite[(4.14)--(4.16)]{RS23}, one has that
\begin{equation}
\label{local_interaction_convergence}
\lim_{N \to \infty} \mcV^N = \mcV
\end{equation}
almost surely. We have that $\chi^2=\chi$. By Corollary \ref{uniform_partition_function_corollary} and Proposition \ref{local_normalisability_lemma}, one has that
\begin{equation*}
\left|\e^{\mc{V}^N} - \e^{\mc{V}}\right|\chi(\mcN) \leq 2 \e^{\frac{1}{3}\|\varphi\|_{L^6}^6} \chi(\mcN) \in L^1(\dd \mu_0).
\end{equation*}
Moreover, by Lemma \ref{random_variable_bound_lemma} and the support properties of $\chi$, it follows that
\begin{equation*}
\Theta(\xi) \chi(\mcN) \in L^\infty(\dd \mu_0)
\end{equation*}
uniformly in $\xi \in \mc{C}_p$. So it follows that 
\begin{equation}
\label{classical_Gibbs_convergence}
\lim_{N\to \infty} \int \dd \mu_0 \left|\e^{\mc{V}^N} - \e^{\mc{V}}\right||\Theta(\xi)|\chi(\mcN) = 0.
\end{equation}
Similarly, one can argue that
\begin{equation}
\label{classical_partition_function_convergence}
\lim_{N \to \infty} z^N = z.
\end{equation}
Writing
\begin{equation*}
\rho^N(\Theta(\xi)) -\rho(\Theta(\xi)) = \frac{1}{z} \int \dd \mu_0 \left(\frac{z}{z^N} \e^{\mcV^N} -e^{\mc{V}}  \right) \Theta(\xi) \chi(\mcN),
\end{equation*}
we deduce that \eqref{delta_2} holds. Here we once again use Corollary \ref{uniform_partition_function_corollary}.
\end{proof}
We are now able to prove Theorem \ref{delta_theorem}.
\begin{proof}[Proof of Theorem \ref{delta_theorem}]
By proving Lemma \ref{state_convergence_delta}, we now have the local analogue of Corollary \ref{state_convergence_corollary}. We can thus proceed via the same duality argument mentioned in the proof of Theorem \ref{bounded_theorem}. We direct the reader to the proofs of \cite[Theorems 1.9 and 1.10]{RS23} for full details, since the proofs are not dependent on the smoothness of the cut-off function.
\end{proof}

\appendix
\section{Proof of Proposition \ref{p-Hamiltonian_uniform_proposition}}
\label{estimates_appendix}
In this section, we prove Proposition \ref{p-Hamiltonian_uniform_proposition}. We will proceed via Wick's theorem. We recall the following lemma, see \cite[Lemma B.1]{FKSS17}.
\begin{lemma}[Quantum Wick theorem]
\label{quantum_Wick_theorem_lemma}
Let $h$ be as in Assumption \ref{free_Hamiltonian_assumption}. Suppose that $\mcA_1,\ldots,\mcA_n$ are operators of the form $\mcA_j = \vph_\eps(f_j)$ or $\mcA_j = \vph^*_\eps(f_j)$ for $f_j \in \mfh$. Then one has that
\begin{equation*}
\frac{1}{Z_{\eps,0}}\mathrm{Tr}(\mcA_1 \ldots \mcA_n \e^{-H_{\eps,0}}) = \sum_{\Pi} \prod_{(i,j) \in \Pi} \frac{1}{Z_{\eps,0}}\mathrm{Tr}(\mcA_i \mcA_j \e^{-H_{\eps,0}}), 
\end{equation*}
where the sum is taken over the set of all pairings of $\{1,\ldots,n\}$ and the edges of $\Pi$ are ordered pairs $(i,j)$ with $i < j$.
\end{lemma}
Before proceeding with the proof of Proposition \ref{p-Hamiltonian_uniform_proposition}, we first prove that following lemma.
\begin{lemma}
\label{trace_class_p_lift_lemma}
Suppose that $\eps > 0$, $p \in \N$, and that $s \in (0,1/2)$. Then one has that
\begin{equation*}
\frac{1}{Z_{\eps,0}} \mathrm{Tr}(\dd \Gamma_\eps((-\Delta)^s)^p \e^{-H_{\eps,0}}) < \infty.
\end{equation*}
\end{lemma}
\begin{proof}
{By Lemma \ref{free_Hamiltonian_1_lemma}, one has that
\begin{equation*}
\frac{1}{Z_{\eps,0}}\mathrm{Tr(\dd \Gamma_\eps((-\Delta)^s) \e^{-\frac{1}{p}H_{\eps,0}})} < \infty.
\end{equation*}
It follows {from the embedding of Schatten $p$-spaces and the fact that
\begin{equation*}
Z_{\eps,0}^{p-1} < C(\eps) < \infty
\end{equation*}
for fixed $\eps \in (0,1)$} that
\begin{equation*}
\frac{1}{Z_{\eps,0}^p}\mathrm{Tr}(\dd \Gamma_\eps((-\Delta)^s)^p \e^{-H_{\eps,0}})) < \infty.
\end{equation*}
Here we use that the two operators in the trace commute and are positive. Then we write
\begin{equation*}
\frac{1}{Z_{\eps,0}} \mathrm{Tr}((\dd \Gamma_\eps(-\Delta)^s)^p\e^{-H_{\eps,0}}) = C(\eps) \frac{1}{Z_{\eps,0}^p}\mathrm{Tr}(\dd \Gamma_\eps((-\Delta)^s)^p \e^{-H_{\eps,0}})).
\end{equation*}}
\end{proof}
We are now able to prove Proposition \ref{p-Hamiltonian_uniform_proposition}.

{\begin{proof}[Proof of Proposition \ref{p-Hamiltonian_uniform_proposition}]
By Lemma \ref{trace_class_p_lift_lemma} and the dominated convergence theorem, it suffices to show that
\begin{equation}
\label{power_p_equation}
\frac{1}{Z_{\eps,0}}\mathrm{Tr}\left(\left[\sum_{|k| \leq N} \lambda_k^s \vph_\eps^*(u_k) \vph_\eps(u_k)\right]^p \e^{-H_{\eps,0}}\right)
\end{equation}
is bounded uniformly in $\eps$. Here we recall $\eqref{eigenvalues}$ and use that $(4\pi^2|k^2|)^{s} + 1 \leq C(s) \lambda_k^{s}$.  We denote by $n_k = a^*(u_k)a(u_k)$, where $a^*$ and $a$ are the non-scaled creation and annihilation operators. Using Lemma \ref{quantum_Wick_theorem_lemma}, \eqref{power_p_equation} can be reduced to the sums (with combinatorial constants that depend on $p)$ of products of terms of the form of
\begin{equation*}
\sum_{|k| \leq N}  \mathrm{Tr}\left((\eps \lambda_k^s)^j (n_k)^j\frac{\e^{-H_\eps,0}}{Z_{\eps,0}}\right).
\end{equation*}
Here we have used that the $u_k$ are orthogonal. So it suffices to bound such terms uniformly in $N$ and $\eps \in (0,1)$. One computes with Lemma \ref{quantum_Wick_theorem_lemma} that up to combinatorial constants that depend on $p$, it suffices to bound
\begin{equation*}
\sum_k (\eps \lambda_k^s)^{i}\nu_{k,\eps}^j < C < \infty
\end{equation*}
uniformly in $\eps$, where $1 \leq j \leq i \leq p$. Here we have defined
\begin{equation*}
\nu_{k,\eps} := \frac{1}{\e^{\eps\lambda_k} -1},
\end{equation*}
{and we have used that
\begin{equation*}
\frac{1}{Z_{\eps,0}}\mathrm{Tr}\left( a^*(g)a(f)\e^{-H_{\eps,0}}\right) = \left\langle f, \frac{1}{\e^{\eps h} -1}g\right\rangle,
\end{equation*}}
see \cite[Lemma B.1]{FKSS17}. We split the sum into high and low frequencies. Recall that $\R \ni \lambda_k \sim |k|^2 + 1$. We consider frequencies such that $\eps \lambda_k \leq p$. We bound 
\begin{multline*}
\sum_{\eps \lambda_k \leq p} (\eps \lambda_k^s)^{i}\nu_{k,\eps}^j \leq \sum_{\eps \lambda_k \leq p} \eps^{i-j} \frac{\lambda_k^{is}}{\lambda_k^j} \lesssim  \eps^{i-j} \sum_{\eps \lambda_k \leq p} \frac{\lambda_k^{is}}{\lambda_k^j} \\
= \eps^{i-j} \sum_{\eps \lambda_k {\leq p}} (1+|k|^2)^{is - j} \lesssim \eps^{i-j}\left(1 + \sum_{\eps k^2 \leq p} k^{2is -2j}\right).
\end{multline*}
If $2is - 2j < -1$, the untruncated sum is finite. Otherwise, one bounds (up to a constant depending on $p$) by
\begin{equation*}
\eps^{i-j} \eps^{-\frac{1}{2} - is +j} \leq \eps^0.
\end{equation*}
Here we have compared with the integral of $1/x^{\alpha}$ up to $K = p\eps^{-\frac{1}{2}}$. The final inequality holds because $s \in (0,1/2)$ and $i \geq 1$. For high frequencies, we consider $\eps \lambda_k > p$. In this case, one bounds
\begin{equation*}
    \sum_{\eps \lambda_k > p} (\eps \lambda_k^s)^{i}\nu_{k,\eps}^j \lesssim \eps^i \sum_{\eps \lambda_k > p} |k|^{2si}\e^{-\frac{j\eps k^2}{2}}.
\end{equation*}
Using that the maximum of the function occurs at
\begin{equation*}
\sqrt{\frac{is}{j\eps}} \leq \eps^{-\frac{1}{2}}\sqrt{\frac{p}{2}}
\end{equation*}
and that
\begin{equation*}
\int x^{m} \e^{-ax^2} \lesssim_{m} \frac{1}{a^{(m+1)/2}},
\end{equation*}
we can bound, up to a constant depending on $j \leq p$, by $\eps^{i}\eps^{-is -\frac{1}{2}}$, which is finite because $s < \frac{1}{2}$ and $i \geq 1$.
\end{proof}}

\subsection*{Acknowledgements}
The authors thank Zied Ammari for pointing them to \cite[Corollary 3.8]{AN11}. S.F.~ acknowledges the support of the Deutsche Forschungsgemeinschaft (DFG, German Research Foundation) through CRC/TRR 352 and the University of Tübingen during her Distinguished Postdoctoral Fellowship. A.R.~ acknowledges the support by the Italian Ministry of University and Research (MUR) through
the PRIN 2022 grant ``OpeN and Effective quantum Systems (ONES)''.
\bibliographystyle{abbrv}
\bibliography{refs}

@article{NierAmmari,
author = {Ammari, Zied and Nier, Francis},
year = {2011},
pages = {},
title = {Mean field propagation of infinite dimensional Wigner measures with a
singular two-body interaction potential},
volume = {14},
journal = {Annali della Scuola normale superiore di Pisa, Classe di scienze},
doi = {10.2422/2036-2145.201112_004}
}

@article{RS22,
 author = {Rout, Andrew and Sohinger, Vedran},
 title = {A microscopic derivation of {Gibbs} measures for the 1D focusing cubic nonlinear {Schr{\"o}dinger} equation},
 fjournal = {Communications in Partial Differential Equations},
 journal = {Commun. Partial Differ. Equations},
 issn = {0360-5302},
 volume = {48},
 number = {7-8},
 pages = {1008--1055},
 year = {2023},
 language = {English},
 doi = {10.1080/03605302.2023.2243491},
 zbMATH = {7755399},
 Zbl = {1528.35168}
}

@article{RS23,
 author = {Rout, Andrew and Sohinger, Vedran},
 title = {A microscopic derivation of {Gibbs} measures for the 1D focusing quintic nonlinear {Schr{\"o}dinger} equation},
 fjournal = {SIAM Journal on Mathematical Analysis},
 journal = {SIAM J. Math. Anal.},
 issn = {0036-1410},
 volume = {57},
 number = {5},
 pages = {4680--4755},
 year = {2025},
 language = {English},
 doi = {10.1137/24M1681227},
 zbMATH = {8096463}
}

@article{AN08,
 author = {Ammari, Zied and Nier, Francis},
 title = {Mean field limit for bosons and infinite dimensional phase-space analysis},
 fjournal = {Annales Henri Poincar{\'e}},
 journal = {Ann. Henri Poincar{\'e}},
 issn = {1424-0637},
 volume = {9},
 number = {8},
 pages = {1503--1574},
 year = {2008},
 language = {English},
 doi = {10.1007/s00023-008-0393-5},
 zbMATH = {5519128},
 Zbl = {1171.81014}
}

@article{FKSS17,
 author = {Fr{\"o}hlich, J{\"u}rg and Knowles, Antti and Schlein, Benjamin and Sohinger, Vedran},
 title = {Gibbs measures of nonlinear {Schr{\"o}dinger} equations as limits of many-body quantum states in dimensions {{\({d {{\leqslant}} 3}\)}}},
 fjournal = {Communications in Mathematical Physics},
 journal = {Commun. Math. Phys.},
 issn = {0010-3616},
 volume = {356},
 number = {3},
 pages = {883--980},
 year = {2017},
 language = {English},
 doi = {10.1007/s00220-017-2994-7},
 zbMATH = {6820551},
 Zbl = {1381.81177}
}

@article{LNR15,
     author = {Lewin, Mathieu and Nam, Phan Th\`anh and Rougerie, Nicolas},
     title = {Derivation of nonlinear {Gibbs} measures from many-body quantum mechanics},
     journal = {Journal de l{\textquoteright}\'Ecole polytechnique - Math\'ematiques},
     pages = {65--115},
     publisher = {Ecole polytechnique},
     volume = {2},
     year = {2015},
     doi = {10.5802/jep.18},
     language = {en},
     url = {https://www.numdam.org/articles/10.5802/jep.18/}
}

@article{AN11,
 author = {Ammari, Zied and Nier, Francis},
 title = {Mean field propagation of {Wigner} measures and {BBGKY} hierarchies for general bosonic states},
 fjournal = {Journal de Math{\'e}matiques Pures et Appliqu{\'e}es. Neuvi{\`e}me S{\'e}rie},
 journal = {J. Math. Pures Appl. (9)},
 issn = {0021-7824},
 volume = {95},
 number = {6},
 pages = {585--626},
 year = {2011},
 language = {English},
 doi = {10.1016/j.matpur.2010.12.004},
 zbMATH = {5909638},
 Zbl = {1251.81062}
}

@article{Soh19,
 author = {Sohinger, Vedran},
 title = {A microscopic derivation of {Gibbs} measures for nonlinear {Schr{\"o}dinger} equations with unbounded interaction potentials},
 fjournal = {IMRN. International Mathematics Research Notices},
 journal = {Int. Math. Res. Not.},
 issn = {1073-7928},
 volume = {2022},
 number = {19},
 pages = {14964--15063},
 year = {2022},
 language = {English},
 doi = {10.1093/imrn/rnab132},
 zbMATH = {7597130},
 Zbl = {1501.35379}
}

@misc{DR24,
 author = {Dinh, Van Duong and Rougerie, Nicolas},
 title = {From bosonic canonical ensembles to non-linear {Gibbs} measures},
 year = {2024},
 howpublished = {Preprint, {arXiv}:2412.13597 [math.{AP}] (2024)},
 url = {https://arxiv.org/abs/2412.13597},
 arXiv = {arXiv:2412.13597}
}

@Article{LNR21,
 Author = {Lewin, Mathieu and Phan Th{\`a}nh Nam and Rougerie, Nicolas},
 Title = {Classical field theory limit of many-body quantum {Gibbs} states in 2D and 3D},
 FJournal = {Inventiones Mathematicae},
 Journal = {Invent. Math.},
 ISSN = {0020-9910},
 Volume = {224},
 Number = {2},
 Pages = {315--444},
 Year = {2021},
 Language = {English},
 DOI = {10.1007/s00222-020-01010-4},
 zbMATH = {7355936},
 Zbl = {1467.82008}
}

@misc{LNR18,
     author = {Lewin, Mathieu and Nam, Phan Th\`anh and Rougerie, Nicolas},
     title = {Classical field theory limit of {2D} many-body quantum {Gibbs} states},
     year = {2015},
      url = {https://arxiv.org/abs/2412.05354},
 arXiv = {arXiv:2412.05354}
}

@article{AFP25,
 author = {Ammari, Zied and Farhat, Shahnaz and Petrat, S{\"o}ren},
 title = {Expansion of the many-body quantum {Gibbs} state of the {Bose}-{Hubbard} model on a finite graph},
 fjournal = {Documenta Mathematica},
 journal = {Doc. Math.},
 issn = {1431-0635},
 volume = {30},
 number = {2},
 pages = {475--496},
 year = {2025},
 language = {English},
 doi = {10.4171/DM/1001},
 zbMATH = {8015944}
}

@misc{NZZ25,
 author = {Phan Th{\`a}nh Nam and Rongchan Zhu and Xiangchan Zhu},
 title = {$\Phi^4_3$ Theory from many-body quantum {Gibbs} states},
 year = {2025},
 howpublished = {Preprint, {arXiv}:2502.04884 [math-ph] (2025)},
 url = {https://arxiv.org/abs/2502.04884},
 arXiv = {arXiv:2502.04884}
}

@misc{LNZ26,
 author = {Lin L{\"u} and Phan Th{\`a}nh Nam and Rongchan Zhu},
 title = {Derivation of the focusing {$\Phi^6_1$} measure in the optimal mass regime from many-body quantum {Gibbs} states},
 year = {2026},
 howpublished = {Preprint, {arXiv}:2605.25755 [math-ph] (2026)},
 url = {https://arxiv.org/abs/2605.25755},
 arXiv = {arXiv:2605.25755}
}

@article{FKSS22,
 author = {Fr{\"o}hlich, J{\"u}rg and Knowles, Antti and Schlein, Benjamin and Sohinger, Vedran},
 title = {The mean-field limit of quantum {Bose} gases at positive temperature},
 fjournal = {Journal of the American Mathematical Society},
 journal = {J. Am. Math. Soc.},
 issn = {0894-0347},
 volume = {35},
 number = {4},
 pages = {955--1030},
 year = {2022},
 language = {English},
 doi = {10.1090/jams/987},
 zbMATH = {7571574},
 Zbl = {1504.35477}
}

@article{FKSS22b,
 author = {Fr{\"o}hlich, J{\"u}rg and Knowles, Antti and Schlein, Benjamin and Sohinger, Vedran},
 title = {The {Euclidean} {{\(\phi_2^4\)}} theory as a limit of an interacting {Bose} gas},
 fjournal = {Journal of the European Mathematical Society (JEMS)},
 journal = {J. Eur. Math. Soc. (JEMS)},
 issn = {1435-9855},
 volume = {27},
 number = {11},
 pages = {4399--4468},
 year = {2025},
 language = {English},
 doi = {10.4171/JEMS/1454},
 zbMATH = {8101462}
}

@article{FKSS18,
 author = {Fr{\"o}hlich, J{\"u}rg and Knowles, Antti and Schlein, Benjamin and Sohinger, Vedran},
 title = {A microscopic derivation of time-dependent correlation functions of the {{\(1D\)}} cubic nonlinear {Schr{\"o}dinger} equation},
 fjournal = {Advances in Mathematics},
 journal = {Adv. Math.},
 issn = {0001-8708},
 volume = {353},
 pages = {67--115},
 year = {2019},
 language = {English},
 doi = {10.1016/j.aim.2019.06.029},
 zbMATH = {7096519},
 Zbl = {1421.82022}
}

@article{OST22,
 author = {Oh, Tadahiro and Sosoe, Philippe and Tolomeo, Leonardo},
 title = {Optimal integrability threshold for {Gibbs} measures associated with focusing {NLS} on the torus},
 fjournal = {Inventiones Mathematicae},
 journal = {Invent. Math.},
 issn = {0020-9910},
 volume = {227},
 number = {3},
 pages = {1323--1429},
 year = {2022},
 language = {English},
 doi = {10.1007/s00222-021-01080-y},
 zbMATH = {7495371},
 Zbl = {1501.35375}
}

@article{Bou93,
 author = {Bourgain, J.},
 title = {Fourier transform restriction phenomena for certain lattice subsets and applications to nonlinear evolution equations. {I}: {Schr{\"o}dinger} equations},
 fjournal = {Geometric and Functional Analysis. GAFA},
 journal = {Geom. Funct. Anal.},
 issn = {1016-443X},
 volume = {3},
 number = {2},
 pages = {107--156},
 year = {1993},
 language = {English},
 doi = {10.1007/BF01896020},
 url = {https://eudml.org/doc/58115},
 zbMATH = {404289},
 Zbl = {0787.35097}
}

@article{FKSS20,
 author = {Fr{\"o}hlich, J{\"u}rg and Knowles, Antti and Schlein, Benjamin and Sohinger, Vedran},
 title = {A path-integral analysis of interacting {Bose} gases and loop gases},
 fjournal = {Journal of Statistical Physics},
 journal = {J. Stat. Phys.},
 issn = {0022-4715},
 volume = {180},
 number = {1-6},
 pages = {810--831},
 year = {2020},
 language = {English},
 doi = {10.1007/s10955-020-02543-x},
 zbMATH = {7239489},
 Zbl = {1446.35203}
}

@book{GJ81,
 author = {Glimm, James and Jaffe, Arthur},
 title = {Quantum physics. {A} functional integral point of view},
 year = {1981},
 language = {English},
 howpublished = {New {York} - {Heidelberg} - {Berlin}: {Springer}-{Verlag}. {XX}, 417 p., 43 ill. \$ 26.40 (1981).},
 zbMATH = {3721546},
 Zbl = {0461.46051}
}

@book{Sim74,
 author = {Simon, Barry},
 title = {The {{\({P}({{\phi}} )_{2}\)}} {Euclidean} (quantum) field theory.},
 fseries = {Princeton Series in Physics},
 series = {Princeton Ser. Phys.},
 year = {1974},
 publisher = {Princeton, NJ: Princeton University Press},
 language = {English},
 zbMATH = {5629250},
 Zbl = {1175.81146}
}

@article{NS19,
 author = {Nahmod, Andrea R. and Staffilani, Gigliola},
 title = {Randomness and nonlinear evolution equations},
 fjournal = {Acta Mathematica Sinica. English Series},
 journal = {Acta Math. Sin., Engl. Ser.},
 issn = {1439-8516},
 volume = {35},
 number = {6},
 pages = {903--932},
 year = {2019},
 language = {English},
 doi = {10.1007/s10114-019-8297-5},
 url = {hdl.handle.net/1721.1/125076},
 zbMATH = {7076577},
 Zbl = {1419.35255}
}

@article{OT18,
 author = {Oh, Tadahiro and Thomann, Laurent},
 title = {A pedestrian approach to the invariant {Gibbs} measures for the 2-{{\(d\)}} defocusing nonlinear {Schr{\"o}dinger} equations},
 fjournal = {Stochastic and Partial Differential Equations. Analysis and Computations},
 journal = {Stoch. Partial Differ. Equ., Anal. Comput.},
 issn = {2194-0401},
 volume = {6},
 number = {3},
 pages = {397--445},
 year = {2018},
 language = {English},
 doi = {10.1007/s40072-018-0112-2},
 zbMATH = {7087785},
 Zbl = {1421.35340}
}

@article{BTT18,
 author = {Burq, Nicolas and Thomann, Laurent and Tzvetkov, Nikolay},
 title = {Remarks on the {Gibbs} measures for nonlinear dispersive equations},
 fjournal = {Annales de la Facult{\'e} des Sciences de Toulouse. Math{\'e}matiques. S{\'e}rie VI},
 journal = {Ann. Fac. Sci. Toulouse, Math. (6)},
 issn = {0240-2963},
 volume = {27},
 number = {3},
 pages = {527--597},
 year = {2018},
 language = {English},
 doi = {10.5802/afst.1578},
 zbMATH = {6979711},
 Zbl = {1405.35193}
}

@article{Bou94,
 author = {Bourgain, J.},
 title = {Periodic nonlinear {Schr{\"o}dinger} equation and invariant measures},
 fjournal = {Communications in Mathematical Physics},
 journal = {Commun. Math. Phys.},
 issn = {0010-3616},
 volume = {166},
 number = {1},
 pages = {1--26},
 year = {1994},
 language = {English},
 doi = {10.1007/BF02099299},
 zbMATH = {722078},
 Zbl = {0822.35126}
}

@article{Bou96,
 author = {Bourgain, Jean},
 title = {Invariant measures for the 2D-defocusing nonlinear {Schr{\"o}dinger} equation},
 fjournal = {Communications in Mathematical Physics},
 journal = {Commun. Math. Phys.},
 issn = {0010-3616},
 volume = {176},
 number = {2},
 pages = {421--445},
 year = {1996},
 language = {English},
 doi = {10.1007/BF02099556},
 zbMATH = {906546},
 Zbl = {0852.35131}
}

@article{Bou97,
 author = {Bourgain, Jean},
 title = {Invariant measures for the {Gross}-{Pitaevskii} equation},
 fjournal = {Journal de Math{\'e}matiques Pures et Appliqu{\'e}es. Neuvi{\`e}me S{\'e}rie},
 journal = {J. Math. Pures Appl. (9)},
 issn = {0021-7824},
 volume = {76},
 number = {8},
 pages = {649--702},
 year = {1997},
 language = {English},
 doi = {10.1016/S0021-7824(97)89965-5},
 zbMATH = {1146373},
 Zbl = {0906.35095}
}

@article{LRS,
 author = {Lebowitz, Joel and Rose, Harvey A. and Speer, Eugene},
 title = {Statistical mechanics of the nonlinear {Schr{\"o}dinger} equation},
 fjournal = {Journal of Statistical Physics},
 journal = {J. Stat. Phys.},
 issn = {0022-4715},
 volume = {50},
 number = {3-4},
 pages = {657--687},
 year = {1988},
 language = {English},
 doi = {10.1007/BF01026495},
 zbMATH = {1226422},
 Zbl = {0925.35142}
}

@article{Zhi91,
 author = {Zhidkov, P. E.},
 title = {On an invariant measure for a nonlinear {Schr{\"o}dinger} equation},
 fjournal = {Soviet Mathematics. Doklady},
 journal = {Sov. Math., Dokl.},
 issn = {0197-6788},
 volume = {43},
 number = {2},
 pages = {431--434},
 year = {1991},
 language = {English},
 zbMATH = {56718},
 Zbl = {0806.35168}
}

@misc{DR23,
 author = {Dinh, Van Duong and Rougerie, Nicolas},
 title = {Invariant {Gibbs} measures for 1D {NLS} in a trap},
 year = {2023},
 howpublished = {Preprint, {arXiv}:2301.02544 [math.{AP}] (2023)},
 url = {https://arxiv.org/abs/2301.02544},
 arXiv = {arXiv:2301.02544}
}

@misc{DRTW23,
 author = {Dinh, Van Duong and Rougerie, Nicolas and Tolomeo, Leonardo and Wang, Yuzhao},
 title = {Statistical mechanics of the radial focusing nonlinear {Schr{\"o}dinger} equation in general traps},
 year = {2023},
 howpublished = {Preprint, {arXiv}:2312.06232 [math.{AP}] (2023)},
 url = {https://arxiv.org/abs/2312.06232},
 arXiv = {arXiv:2312.06232}
}

@article{BS96,
 author = {Brydges, David C. and Slade, Gordon},
 title = {Statistical mechanics of the 2-dimensional focusing nonlinear {Schr{\"o}dinger} equation},
 fjournal = {Communications in Mathematical Physics},
 journal = {Commun. Math. Phys.},
 issn = {0010-3616},
 volume = {182},
 number = {2},
 pages = {485--504},
 year = {1996},
 language = {English},
 doi = {10.1007/BF02517899},
 zbMATH = {1002482},
 Zbl = {0867.35090}
}

@article{LW22,
 author = {Liang, Rui and Wang, Yuzhao},
 title = {Gibbs measure for the focusing fractional {NLS} on the torus},
 fjournal = {SIAM Journal on Mathematical Analysis},
 journal = {SIAM J. Math. Anal.},
 issn = {0036-1410},
 volume = {54},
 number = {6},
 pages = {6096--6118},
 year = {2022},
 language = {English},
 doi = {10.1137/21M1445946},
 zbMATH = {7620347},
 Zbl = {1502.35154}
}

@article{Nel73,
 author = {Nelson, Edward},
 title = {The free {Markoff} field},
 fjournal = {Journal of Functional Analysis},
 journal = {J. Funct. Anal.},
 issn = {0022-1236},
 volume = {12},
 pages = {211--227},
 year = {1973},
 language = {English},
 doi = {10.1016/0022-1236(73)90025-6},
 zbMATH = {3429933},
 Zbl = {0273.60079}
}

@misc{Nel73b,
 author = {Nelson, Edward},
 title = {Probability theory and {Euclidean} field theory},
 year = {1973},
 language = {English},
 howpublished = {Constr. {Quant}. {Field} {Theory}, {Lect}. {Notes} {Phys}. 25, 94-124 (1973).},
 zbMATH = {3573034},
 Zbl = {0367.60108}
}

@article{AFS24,
 author = {Ammari, Zied and Farhat, Shahnaz and Sohinger, Vedran},
 title = {Almost sure existence of global solutions for general initial value problems},
 fjournal = {Advances in Mathematics},
 journal = {Adv. Math.},
 issn = {0001-8708},
 volume = {453},
 pages = {61},
 note = {Id/No 109805},
 year = {2024},
 language = {English},
 doi = {10.1016/j.aim.2024.109805},
 zbMATH = {7915891},
 Zbl = {1548.35016}
}

@article{CFL16,
 author = {Carlen, Eric A. and Fr{\"o}hlich, J{\"u}rg and Lebowitz, Joel},
 title = {Exponential relaxation to equilibrium for a one-dimensional focusing non-linear {Schr{\"o}dinger} equation with noise},
 fjournal = {Communications in Mathematical Physics},
 journal = {Commun. Math. Phys.},
 issn = {0010-3616},
 volume = {342},
 number = {1},
 pages = {303--332},
 year = {2016},
 language = {English},
 doi = {10.1007/s00220-015-2511-9},
 zbMATH = {6545571},
 Zbl = {1341.35148}
}

@article{BDNY24,
 author = {Bringmann, Bjoern and Deng, Yu and Nahmod, Andrea R. and Yue, Haitian},
 title = {Invariant {Gibbs} measures for the three dimensional cubic nonlinear wave equation},
 fjournal = {Inventiones Mathematicae},
 journal = {Invent. Math.},
 issn = {0020-9910},
 volume = {236},
 number = {3},
 pages = {1133--1411},
 year = {2024},
 language = {English},
 doi = {10.1007/s00222-024-01254-4},
 zbMATH = {7846295},
 Zbl = {1541.35502}
}

@article{DNY24,
 author = {Deng, Yu and Nahmod, Andrea R. and Yue, Haitian},
 title = {Invariant {Gibbs} measures and global strong solutions for nonlinear {Schr{\"o}dinger} equations in dimension two},
 fjournal = {Annals of Mathematics. Second Series},
 journal = {Ann. Math. (2)},
 issn = {0003-486X},
 volume = {200},
 number = {2},
 pages = {399--486},
 year = {2024},
 language = {English},
 doi = {10.4007/annals.2024.200.2.1},
 zbMATH = {7995674},
 Zbl = {1559.35324}
}

@article{Bri24,
 author = {Bringmann, Bjoern},
 title = {Invariant {Gibbs} measures for the three-dimensional wave equation with a {Hartree} nonlinearity. {II}: {Dynamics}},
 fjournal = {Journal of the European Mathematical Society (JEMS)},
 journal = {J. Eur. Math. Soc. (JEMS)},
 issn = {1435-9855},
 volume = {26},
 number = {6},
 pages = {1933--2089},
 year = {2024},
 language = {English},
 doi = {10.4171/JEMS/1317},
 zbMATH = {7855540},
 Zbl = {1541.60042}
}

@article{Bri22,
 author = {Bringmann, Bjoern},
 title = {Invariant {Gibbs} measures for the three-dimensional wave equation with a {Hartree} nonlinearity. {I}: measures},
 fjournal = {Stochastics and Partial Differential Equations. Analysis and Computations},
 journal = {Stoch. Partial Differ. Equ., Anal. Comput.},
 issn = {2194-0401},
 volume = {10},
 number = {1},
 pages = {1--89},
 year = {2022},
 language = {English},
 doi = {10.1007/s40072-021-00193-y},
 zbMATH = {7507356},
 Zbl = {1491.60095}
}

@article{DNY22,
 author = {Deng, Yu and Nahmod, Andrea R. and Yue, Haitian},
 title = {Random tensors, propagation of randomness, and nonlinear dispersive equations},
 fjournal = {Inventiones Mathematicae},
 journal = {Invent. Math.},
 issn = {0020-9910},
 volume = {228},
 number = {2},
 pages = {539--686},
 year = {2022},
 language = {English},
 doi = {10.1007/s00222-021-01084-8},
 zbMATH = {7514024},
 Zbl = {1506.35208}
}

@book{OTT24,
 author = {Oh, Tadahiro and Okamoto, Mamoru and Tolomeo, Leonardo},
 title = {Focusing {{\(\Phi^4_3\)}}-model with a {Hartree}-type nonlinearity},
 fseries = {Memoirs of the American Mathematical Society},
 series = {Mem. Am. Math. Soc.},
 issn = {0065-9266},
 volume = {1529},
 isbn = {978-1-4704-7193-4; 978-1-4704-8008-0},
 year = {2024},
 publisher = {Providence, RI: American Mathematical Society (AMS)},
 language = {English},
 doi = {10.1090/memo/1529},
 zbMATH = {7975522},
 Zbl = {1558.35006}
}

@article{GH21,
 author = {Gubinelli, Massimiliano and Hofmanov{\'a}, Martina},
 title = {A {PDE} construction of the {Euclidean} {{\(\Phi^4_3\)}} quantum field theory},
 fjournal = {Communications in Mathematical Physics},
 journal = {Commun. Math. Phys.},
 issn = {0010-3616},
 volume = {384},
 number = {1},
 pages = {1--75},
 year = {2021},
 language = {English},
 doi = {10.1007/s00220-021-04022-0},
 zbMATH = {7348133},
 Zbl = {1514.81190}
}

@article{BG20,
 author = {Barashkov, N. and Gubinelli, M.},
 title = {A variational method for {{\(\Phi^4_3\)}}},
 fjournal = {Duke Mathematical Journal},
 journal = {Duke Math. J.},
 issn = {0012-7094},
 volume = {169},
 number = {17},
 pages = {3339--3415},
 year = {2020},
 language = {English},
 doi = {10.1215/00127094-2020-0029},
 zbMATH = {7292332},
 Zbl = {1508.81928}
}

@misc{ARS24,
 author = {Ammari, Zied and Rout, Andrew and Sohinger, Vedran},
 title = {Gibbs measures as local equilibrium {KMS} states for focusing nonlinear {Schr{\"o}dinger} equations},
 year = {2024},
 howpublished = {Preprint, {arXiv}:2412.05354 [math.{AP}] (2024)},
 url = {https://arxiv.org/abs/2412.05354},
 arXiv = {arXiv:2412.05354}
}

@article{AS23,
 author = {Ammari, Zied and Sohinger, Vedran},
 title = {Gibbs measures as unique {KMS} equilibrium states of nonlinear {Hamiltonian} {PDEs}},
 fjournal = {Revista Matem{\'a}tica Iberoamericana},
 journal = {Rev. Mat. Iberoam.},
 issn = {0213-2230},
 volume = {39},
 number = {1},
 pages = {29--90},
 year = {2023},
 language = {English},
 doi = {10.4171/RMI/1366},
 zbMATH = {7686640},
 Zbl = {1514.35403}
}

@article{RSTW25,
 author = {Robert, Tristan and Seong, Kihoon and Tolomeo, Leonardo and Wang, Yuzhao},
 title = {Focusing {Gibbs} measures with harmonic potential},
 fjournal = {Annales de l'Institut Henri Poincar{\'e}. Probabilit{\'e}s et Statistiques},
 journal = {Ann. Inst. Henri Poincar{\'e}, Probab. Stat.},
 issn = {0246-0203},
 volume = {61},
 number = {1},
 pages = {571--598},
 year = {2025},
 language = {English},
 doi = {10.1214/23-AIHP1435},
 zbMATH = {8013071},
 Zbl = {1560.60110}
}

@article{OQ13,
 author = {Oh, Tadahiro and Quastel, Jeremy},
 title = {On invariant {Gibbs} measures conditioned on mass and momentum},
 fjournal = {Journal of the Mathematical Society of Japan},
 journal = {J. Math. Soc. Japan},
 issn = {0025-5645},
 volume = {65},
 number = {1},
 pages = {13--35},
 year = {2013},
 language = {English},
 doi = {10.2969/jmsj/06510013},
 zbMATH = {6152426},
 Zbl = {1274.60214}
}

@article{BTT13,
 author = {Burq, Nicolas and Thomann, Laurent and Tzvetkov, Nikolay},
 title = {Long time dynamics for the one dimensional non linear {Schr{\"o}dinger} equation},
 fjournal = {Annales de l'Institut Fourier},
 journal = {Ann. Inst. Fourier},
 issn = {0373-0956},
 volume = {63},
 number = {6},
 pages = {2137--2198},
 year = {2013},
 language = {English},
 doi = {10.5802/aif.2825},
 zbMATH = {6325429},
 Zbl = {1317.35226}
}

@misc{JR25,
 author = {Jougla, Lucas and Rougerie, Nicolas},
 title = {$\Phi^4_2$ theory limit of a many-body bosonic free energy},
 year = {2025},
 howpublished = {Preprint, {arXiv}:2512.10704 [math.{AP}] (2025)},
 url = {https://arxiv.org/abs/2512.10704},
 arXiv = {arXiv:2512.10704}
}

@article{FKSS23,
 author = {Fr{\"o}hlich, J{\"u}rg and Knowles, Antti and Schlein, Benjamin and Sohinger, Vedran},
 title = {Interacting loop ensembles and {Bose} gases},
 fjournal = {Annales Henri Poincar{\'e}},
 journal = {Ann. Henri Poincar{\'e}},
 issn = {1424-0637},
 volume = {24},
 number = {5},
 pages = {1439--1503},
 year = {2023},
 language = {English},
 doi = {10.1007/s00023-022-01238-1},
 url = {wrap.warwick.ac.uk/168693/1/WRAP-Interacting-loop-ensembles-Bose-gases-23.pdf},
 zbMATH = {7681332}
}

@misc{Nik25,
 author = {Niko Nikov},
 title = {Phase transitions for fractional {$\Phi^3_d$} on the torus},
 year = {2025},
 howpublished = {Preprint, {arXiv}:2501.17669 [math.{PR}] (2025)},
 url = {https://arxiv.org/abs/2501.17669},
 arXiv = {arXiv:2501.17669}
}

@misc{TV26,
 author = {Arnaud Triay and Fran{\c{c}}ois L. A. Visconti},
 title = {Derivation of the {3D} quintic {Gross}--{Pitaevskii} equation},
 year = {2026},
 howpublished = {Preprint, {arXiv}:2602.06481 [math-ph] (2026)},
 url = {https://arxiv.org/abs/2602.06481},
 arXiv = {arXiv:2602.06481}
}

@misc{CKRTG26,
 author = {Cristina Caraci and Antti Knowles and Alessio Ranallo and Pedro Torres Giesteira},
 title = {The {Euclidean} {$\varphi^4_2$} theory as a limit of an inhomogeneous {Bose} gas},
 year = {2026},
 howpublished = {Preprint, {arXiv}:2603.12241 [math-ph] (2026)},
 url = {https://arxiv.org/abs/2603.12241},
 arXiv = {arXiv:2603.12241}
}

@misc{GSS26,
 author = {Spyros Garouniatis and Grega Saksida and Vedran Sohinger},
 title = {The large-mass limit of interacting quantum gases in the continuum},
 year = {2026},
 howpublished = {Preprint, {arXiv}:2604.23020 [math-ph] (2026)},
 url = {https://arxiv.org/abs/2604.23020},
 arXiv = {arXiv:2604.23020}
}

@misc{NZZ26,
 author = {Phan Th{\`a}nh Nam and Rongchan Zhu and Xiangchan Zhu},
 title = {Derivation of {Gibbs} measure from {Gibbs} state with the fractional {Bessel} interaction in {Two} {Dimensions}},
 year = {2026},
 howpublished = {Preprint, {arXiv}:2604.21583 [math-ph] (2026)},
 url = {https://arxiv.org/abs/2604.21583},
 arXiv = {arXiv:2604.21583}
}

@book{BR81,
 author = {Bratteli, Ola and Robinson, Derek W.},
 title = {Operator algebras and quantum statistical mechanics. {II}: {Equilibrium} states. {Models} in quantum statistical mechanics},
 fseries = {Texts and Monographs in Physics},
 series = {Texts Monogr. Phys.},
 issn = {0172-5998},
 year = {1981},
 publisher = {Springer, New York, NY},
 language = {English},
 zbMATH = {3725098},
 Zbl = {0463.46052}
}

\end{document}